\documentclass[12pt, draftclsnofoot, onecolumn]{IEEEtran}
\usepackage{multicol}
\usepackage{amsmath,epsfig,comment}
\usepackage{amsthm}
\usepackage{tcolorbox}
\usepackage{amssymb}
\usepackage{multirow}
\usepackage{mathrsfs}
\usepackage{amsmath}
\usepackage{arydshln}
\usepackage{multirow}
\usepackage{arydshln}

\usepackage{cases} 
\usepackage{amssymb,amsmath,cite}
\usepackage{epsfig}
\usepackage{color}
\usepackage{bm}
\usepackage{graphicx,subfigure}
\usepackage{algorithm}
\usepackage{algorithmic}
\usepackage{multirow}
\usepackage{soul}

\usepackage{extarrows}

\def\s{\mathbf{s}}

\def\ba{\mathbf{a}}
\def\bC{\mathbf{C}}
\def\bc{\mathbf{c}}

\def\I{\mathcal{I}}
\def\RR{\mathcal{R}}

\def\n{\mathbf{n}}

\def\I{\mathcal{I}}
\def\R{\mathbb{R}}

\def\C{\mathbb{C}}
\def\bb{\mathbf{b}}
\def\bW{\mathbf{W}}
\def\bV{\mathbf{V}}
\def\bU{\mathbf{U}}
\def\bQ{\mathbf{Q}}

\def\H{\mathbf{H}}

\def\bP{\mathbf{P}}

\def\y{\mathbf{y}}

\def\x{\mathbf{x}}
\def\s{\mathbf{s}}

\def\h{\mathbf{h}}

\def\bA{\mathbf{A}}
\def\bM{\mathbf{M}}

\def\bD{\mathbf{D}}
\def\bR{\mathbf{R}}
\def\br{\mathbf{r}}

\def\bH{\mathbf{H}}
\def\bB{\mathbf{B}}
\def\bG{\mathbf{G}}

\def\bthe{\boldsymbol{\Theta}}

\def\bw{\mathbf{w}}
\def\bd{\mathbf{d}}

\def\v{\mathbf{v}}
\def\bw{\mathbf{w}}
\def\bY{\mathbf{Y}}

\def\a{\boldsymbol{\Lambda}}
\newtheorem{theorem}{Theorem}

\newtheorem{lemma}{Lemma}
\newtheorem{remark}{Remark}

\newtheorem{definition}{Definition}

\newtheorem{corollary}{Corollary}

\newtheoremstyle{noparens}%
  {}{}%
  {\itshape}{}%
  {\bfseries}{.}%
  { }%
  {\thmname{#1}\thmnumber{ #2}\mdseries\thmnote{ #3}}

\theoremstyle{noparens}

\title{Beyond-Diagonal RIS in Multiuser MIMO: Graph Theoretic Modeling and Optimal Architectures with Low Complexity}
\author{\IEEEauthorblockN{Zheyu Wu and Bruno Clerckx}
  	\thanks{Z. Wu and and B. Clerckx are with the Department of Electrical and Electronic Engineering, Imperial College London, London, SW7 2AZ, U.K. (email: \{zheyu.wu, b.clerckx\}@imperial.ac.uk).  Bruno Clerckx is also with Kyung Hee University, Seoul, Korea.  This work has been partially supported by UKRI grant EP/Y004086/1, EP/X040569/1, EP/Y037197/1, EP/X04047X/1, EP/Y037243/1 (\emph{Corresponding Author: Bruno Clerckx}).
 	}
  }
\date{today}
\begin{document}
\maketitle
\begin{abstract}
Reconfigurable intelligent surfaces (RIS) is regarded as a key enabler of wave/analog-domain beamforming, processing, and computing in future wireless communication systems.  Recently, Beyond-Diagonal RIS (BD-RIS) has been proposed as a generalization of conventional RIS, offering enhanced design flexibility thanks to the presence of tunable impedances that connect RIS elements. However, increased interconnections lead to high circuit complexity, which poses a significant practical challenge. In this paper, we  address the fundamental open question: ``What is the class of BD-RIS architectures that achieves the optimal performance in a RIS-aided multiuser multi-input multi-output (MIMO) system?" By modeling BD-RIS architectures using graph theory, we identify a class of BD-RIS architectures that achieves the optimal performance—matching that of fully-connected RIS—while maintaining low circuit complexity. Our result holds for a broad class of performance metrics, including the commonly used sum channel gain/sum-rate/energy efficiency maximization, transmit power minimization, and the information-theoretic capacity region. The number of tunable impedances in the proposed class is $\mathcal{O}(N_I\min\{D,N_I/2\})$, where $N_I$ denotes the number of RIS elements and $D$ is the degree of freedom  of the multiuser MIMO channel, i.e., the minimum between the number of transmit antennas and the total number of received antennas across all users. Since $D$ is much smaller than $N_I$ in practice, the complexity scales as $\mathcal{O}(N_ID)$, which is substantially lower than the $\mathcal{O}(N_I^2)$ complexity of fully-connected RIS.  We further introduce two novel BD-RIS architectures—band-connected RIS and stem-connected RIS—and show that they belong to the optimal architecture class under certain conditions. Simulation results validate the optimality and enhanced performance-complexity tradeoff of our proposed architectures and demonstrate the significant loss \textcolor{black}{and limitations} incurred by conventional RIS  compared to BD-RIS in multiuser MIMO settings. 

 \end{abstract}\vspace{-0.1cm}
\begin{IEEEkeywords}
Beyond diagonal reconfigurable intelligent surface, graph theory, low-complexity architecture, multiuser multi-input multi-output system.
\end{IEEEkeywords}
\section{Introduction}
The sixth-generation (6G) wireless network is envisioned to support future use cases and applications that demand  high data rates,  wide coverage, and low latency. These requirements  drive the development of new advanced and intelligent technologies \cite{6G}.  Reconfigurable intelligent surface (RIS) --- a planar surface composed of a large number of  passive reflective elements --- has emerged as a key enabling technology for 6G evolution to enable beamforming, processing, and even computing directly in the analog radio-frequency domain (also known as wave domain). In the RIS paradigm, each reflective element induces a phase shift to collaboratively steer the incident electromagnetic
wave toward the desired direction. Due to the near-passive nature of the reflective elements, RIS is able to operate  in a cost-effective and energy-efficient way.   RIS has garnered significant research interest over the past five years, and its potential has been extensively explored in the existing literature \cite{RIS1,RIS_survey,RIS_mag}.  

Very recently, the concept of Beyond-Diagonal RIS (BD-RIS) has been introduced in \cite{BDRIS} as an innovative generalization of conventional RIS \textcolor{black}{(also known as diagonal RIS or single-connected RIS)}. Specifically, the authors in \cite{BDRIS} studied the modeling of  RIS  based on scattering parameters within the framework of  multiport network analysis. In conventional RIS, each reflective element is connected only to its own reconfigurable load \textcolor{black}{and not to other RIS elements (hence the single-connected RIS terminology)}, resulting in a diagonal scattering matrix \textcolor{black}{(hence the diagonal RIS terminology)}. To achieve additional flexibility, BD-RIS allows for interconnection between reflective elements through tunable impedances. As a result, the scattering matrix is no longer limited to being diagonal, offering enhanced wave manipulation capability \cite{BDRIS, BDRIS_survey}.  

Existing studies have demonstrated that BD-RIS can achieve significantly enhanced performance  compared to conventional RIS \cite{BDRIS,BDRIS_survey, closeform}. Furthermore, by allowing waves absorbed by one element to flow through other elements, BD-RIS supports multiple operational modes, including reflective, transmissive, and hybrid modes \cite{group_conn}. In particular, the hybrid mode enables simultaneous signal reflection and transmission, thereby achieving full-space coverage \cite{group_conn}.  As a generalization of hybrid modes, the multi-sector mode has been proposed in \cite{coverage} to deliver further performance improvements. 
Note that in most existing works, the BD-RIS architecture is assumed to be reciprocal and is characterized by a symmetric scattering matrix. However, BD-RIS is actually a broader concept that also accommodates non-reciprocal architectures with non-symmetric scattering matrices. The potential of non-reciprocal BD-RIS has been recently investigated in \cite{attack,duplex,liu_nonreciprocal,Linonreciprocal2022}.

Despite the great flexibility offered by BD-RIS, its high circuit complexity poses a significant practical challenge. In fully-connected RIS, where every pair of reflective elements is connected, the total number of required impedances scales quadratically with the number of RIS elements, which becomes prohibitive as the number of RIS elements  grows large. Various BD-RIS architectures have been proposed to address this concern. The group-connected RIS partitions the RIS elements into several groups, where only elements in the same group are connected \cite{BDRIS}.  It makes a balance between \textcolor{black}{conventional} single-connected RIS and fully-connected RIS, achieving a better  performance-complexity trade-off. The performance of group-connected RIS can be further improved with appropriate dynamic grouping strategies \cite{grouping, grouping2}. In \cite{tree}, the authors modeled the BD-RIS architecture using graph theory and proposed two new architectures\,---\,tree-connected RIS and forest-connected RIS. A remarkable theoretical result in \cite{tree} is that, in single-user multi-input single-output (MISO) systems, the tree-connected RIS is the optimal architecture that achieves the best performance with the least circuit complexity.  Furthermore,  forest-connected RIS has been proven to achieve the  Pareto frontier between performance and circuit complexity in  single-user MISO systems\cite{forest}. However, the aforementioned desirable properties of tree- and forest-connected architectures do not extend to multi-input multi-output  (MIMO) and/or multiuser systems \cite{wu}. To enhance the performance of tree-connected RIS in multiuser systems while maintaining low circuit complexity, the authors in \cite{qstem} proposed a novel architecture called stem-connected RIS. 
 It has been demonstrated numerically that  the stem-connected RIS attains the sum channel gain comparable to that  achieved by fully-connected RIS. However, there is no theoretical guarantee regarding the performance of stem-connected RIS. 

To date, it remains an open problem what is the optimal BD-RIS architecture achieving the best performance with the least circuit complexity in MIMO and/or multiuser systems. In this paper, we make significant progress on this theoretical problem: we establish a general sufficient condition under which the BD-RIS architecture achieves the best performance, i.e., the same performance as fully-connected RIS, in multiuser MIMO systems. Our result generalizes existing analysis for single-user MISO systems in \cite{tree} and provides a rigorous theoretical justification for the good numerical performance of stem-connected RIS observed in \cite{qstem}. The main contributions are summarized as follows.

First, we model the BD-RIS architecture using graph theory as in \cite{tree} and characterize the topological connectivity of the BD-RIS circuit through the adjacency matrix of the graph. To assess the performance of the BD-RIS architecture, we formulate a general utility optimization problem that encompasses all existing models (e.g., sum-rate maximization \cite{wu,PDD}, transmit power minimization \cite{PDD,wu}, energy efficiency maximization \cite{PDD}, sum channel gain maximization \cite{qstem}) in the literature as special cases. By focusing on this general model, we ensure the universality of our derived theoretical result, which holds regardless of  the performance metric.

Second, based on the above modeling, we provide a general condition on the adjacency matrix under which the corresponding BD-RIS architecture matches the performance of fully-connected RIS. The circuit complexity, i.e., the number of tunable impedances, for the proposed architecture class is  $\mathcal{O}(N_I\min\{D,N_I/2\})$, where $N_I$ denotes the number of RIS elements and $D=\min\{\sum_{k=1}^K N_k, N_T\}$, with $K$, $N_k$,  and $N_T$ denoting the number of users, the number of antennas at the $k$-th user, and the number of antennas at the transmitter, respectively, i.e., $D$ is the degree of freedom (DoF) of the multiuser MIMO channel. {\color{black}This suggests that  to achieve full flexibility of BD-RIS in practical systems (where $N_I\gg D$), it suffices that the number of tunable impedances scales as $\mathcal{O}(N_ID)$, which is significantly lower than the $\mathcal{O}(N_I^2)$ complexity of fully-connected RIS. 
Our result  holds for a broad class of performance metrics, including the commonly used sum channel gain/sum-rate/energy efficiency maximization, transmit power minimization, and the information-theoretic capacity region.}
To the best of our knowledge, the proposed architecture class is the first with a theoretical performance guarantee in multiuser MIMO systems. In particular, when $K=1$ and $N_k=1$, the proposed architecture class reduces to the tree-connected RIS, which coincides with the results in \cite{tree}.

 Third, we introduce two novel BD-RIS architectures: the band-connected RIS and the stem-connected RIS \cite{qstem}, which are characterized by the band and stem widths, respectively. We prove that when the band and stem widths are $2\min\{D,N_I/2\}-1$, both architectures belong to the proposed optimal architecture class, which justifies the observation in \cite{qstem}. In addition, our numerical results show that, by ranging the band and stem widths from $0$ (which corresponds to the conventional single-connected RIS) to $2\min\{D,N_I/2\}-1$, the two architectures achieve a better performance-complexity trade-off than group-connected RIS, demonstrating superior performance at the same circuit complexity. 
Simulation results also demonstrate the significant benefits of BD-RIS over conventional RIS in multiuser  MIMO systems and the inability of conventional RIS to exploit the presence of multiple users to boost the sum-rate.

\emph{Organization}: The remaining parts of this paper are organized as follows. In Section \ref{sec:2}, we introduce the graph theoretical modeling of BD-RIS architecture, the BD-RIS-aided communication model,  and the problem formulation. In Section \ref{sec:3}, we present our main theoretical result,  namely Theorem \ref{theorem1}, which provides a general condition under which the BD-RIS architecture achieves the best performance. Additionally, two novel BD-RIS architectures are introduced as illustrative examples that satisfy Theorem \ref{theorem1}.    The proof of Theorem \ref{theorem1} is provided in Section \ref{sec:4}. Simulation results are given in Section \ref{sec:5} to validate our theoretical result and the paper is concluded in Section \ref{sec:6}. 

\emph{Notations}: We use boldface lower-case letters, boldface upper-case letters, and upper-case calligraphic letters to represent vectors, matrices, and sets, respectively. For a vector $\x$, $x_n$ refers to its $n$-th entry. 
For a matrix $\mathbf{X}$, $\mathbf{X}(i_1:i_2,j_1:j_2)$ denotes the submatrix of $\mathbf{X}$ formed by the rows indexed from $i_1$ to $i_2$ and  the columns indexed from $j_1$ to $j_2$, and $\mathbf{X}(i_1:i_2,:)$ denotes the submatrix formed by the rows indexed from $i_1$ to $i_2$ and all columns. In particular, $\mathbf{X}(m,n)$ represents the $(m,n)$-th entry of $\mathbf{X}$. To ease the notation, $X_{m,n}$ is also used in certain contexts if it does not cause any confusion. The operators $\det\mathbf{X}$ denotes the determinant of $\mathbf{X}$. The operators $(\cdot)^T$, $(\cdot)^H$, $\text{conj}\,(\cdot)$,  $(\cdot)^{-1}$, $\mathcal{R}(\cdot)$, and $\mathcal{I}(\cdot)$ take transpose, conjugate transpose, conjugate, inverse, the real part, and the imaginary part of their corresponding arguments, respectively. The symbols $\mathbb{C}$ and $\mathbb{R}$ denote the complex space and  the real space, respectively. The symbol $\mathbf{I}$ represents the identity matrix with an appropriate dimension. If necessary, the dimension is specified as a subscript. {\color{black}The symbol $\mathbf{e}_l\in\R^{n}$ is a standard basis vector in which only the $l$-th element is $1$, and all other elements are $0$.} Finally, $\mathsf{i}$ is the imaginary unit.  

\section{System Model and Problem Formulation}\label{sec:2}
In this section, we introduce the graph theoretical modeling of BD-RIS architectures,  the BD-RIS-aided multiuser MIMO system model, and the mathematical formulation of our considered problem.

\subsection{BD-RIS modeling}\label{existing_arch}
According to \cite{BDRIS}, an $N_I$-element BD-RIS can be modeled as $N_I$ antennas connected to an $N_I$-port reconfigurable impedance network.  To characterize the circuit topology of BD-RIS, it is convenient to describe the reconfigurable impedance network in terms of its admittance matrix $\bY\in\C^{N_I\times N_I}$, i.e., the $Y$-parameter \cite{microwavebook}. Specifically, let $\bar{Y}_n$ and $\bar{Y}_{n,m}$ denote the self admittance of the $n$-th port and the admittance connecting the $n$-th and $m$-th ports, respectively, then the admittance matrix $\mathbf{Y}$ of the $N_I$-port reconfigurable impedance network is given by \cite{BDRIS}
$$
Y_{n,m}=\left\{
\begin{aligned}
&-\bar{Y}_{n,m},~~~~~~~~~\,~~&\text{if}~n\neq m;\\
&\bar{Y_n}+\textstyle\sum_{j\neq n}\bar{Y}_{n,j},~~&\text{if }n=m.
\end{aligned}\right.
$$
 Note that $Y_{n,m}\neq 0$ if and only if $\bar{Y}_{n,m}\neq 0$. Hence,   interconnections between ports can be represented by nonzero elements in $\bY$. Under the common assumptions that the reconfigurable impedance network is lossless and reciprocal\footnote{In the lossy and non-reciprocal case, the BD-RIS architecture can be modeled as a weighted directed graph. The design of lossy and non-reciprocal BD-RIS architecture is left as future work. } \cite{generalmodel,microwavebook}, the admittance matrix $\bY$ should be purely imaginary and symmetric, which we denote as 
  \begin{equation}\label{YandB}
  \bY=\mathsf{i} \bB,~\text{where }\bB=\bB^T.
  \end{equation} 
  Here, $\bB\in\R^{N_I\times N_I}$ is the susceptance matrix of the reconfigurable impedance network.  

Following \cite{generalmodel}, we model the circuit topology of BD-RIS and the corresponding admittance matrix $\bY$ using a graph 
$\mathcal{G}=(\mathcal{V},\mathcal{E}),$
where the vertex set $\mathcal{V}=\{V_1,V_2,\dots, V_{N_I}\}$ corresponds to the $N_I$ ports of the reconfigurable impedance network,  and the edge set $\mathcal{E}=\left\{(V_n,V_m)\mid {Y}_{n,m}\neq 0\right\}$ captures  all the interconnections between ports. Depending on the circuit topology of the reconfigurable impedance network, existing BD-RIS architectures can be categorized into various types.
\begin{figure}
\includegraphics[width=1\columnwidth]{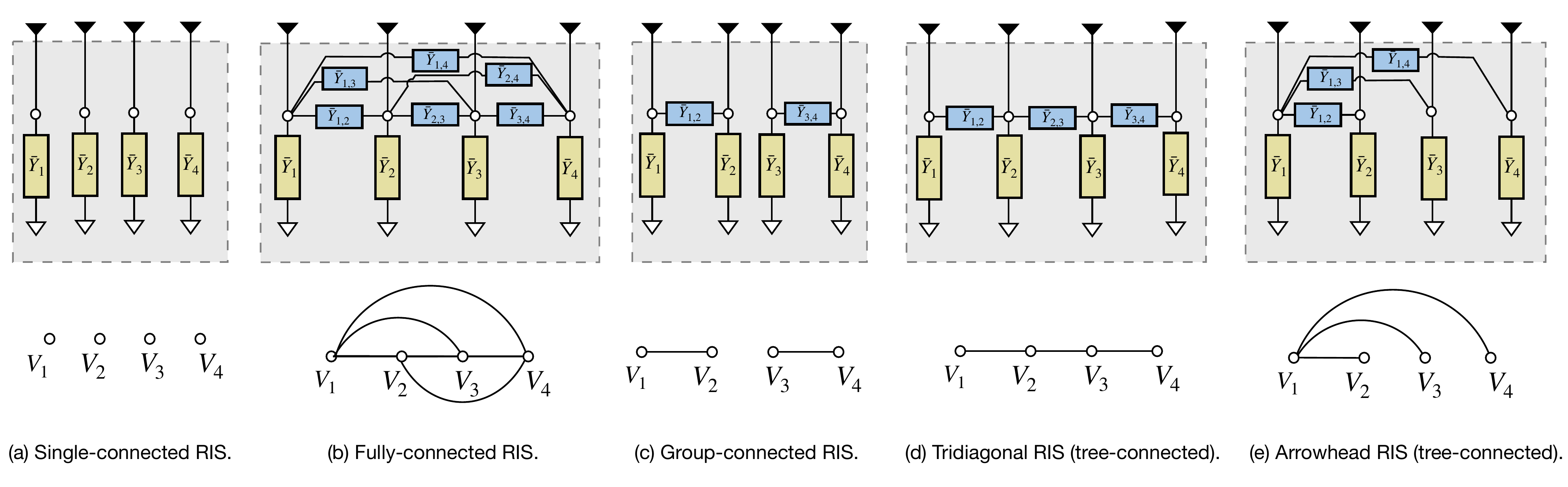}
\centering
\caption{Different RIS architectures and their graph representation.}
\label{ris_arch}
\end{figure}

\begin{itemize}
\item \emph{Single-connected RIS (also known as conventional RIS or diagonal RIS)}.  In single-connected RIS, there is no interconnection between ports. Hence, $\bY$ is a  diagonal matrix, and the corresponding graph is an empty graph \cite{bondy2008graph}, i.e., $\mathcal{E}=\emptyset$; see Fig. \ref{ris_arch} (a).
\item \emph{Fully-connected RIS}. In fully-connected RIS, every pair of ports is connected, and thus $\bY$ is a full matrix. The corresponding graph is a complete graph \cite{bondy2008graph}, i.e., $\mathcal{E}=\{(V_n,V_m)\mid 1\leq n, m\leq N_I, ~n\neq m\}$; see Fig. \ref{ris_arch} (b).
\item \emph{Group-connected RIS}. In group-connected RIS, the ports of the reconfigurable impedance network are  divided into several independent groups. Within each group, fully-connected architecture is employed. In this scheme,   $\bY$ exhibits a  block diagonal structure and the corresponding graph is a cluster graph \cite{bondy2008graph}, i.e., a disjoint union of complete graphs. Specifically, let $G$  denote the number of groups, then $\mathcal{V}=\mathcal{V}_1\cup\mathcal{V}_2\cup\cdots\cup\mathcal{V}_{G}$ with $\mathcal{V}_i\cap\mathcal{V}_j=\emptyset$ for all $i\neq j$, {\color{black} where $\mathcal{V}_g$ represents the vertices in the $g$-th cluster, corresponding to the ports in the $g$-th group,~$g=1,2,\dots, G$.  The edge set $\mathcal{E}$ is defined as $\mathcal{E}=\mathcal{E}_1\cup\mathcal{E}_2\cup\cdots\cup\mathcal{E}_{G}$, where $\mathcal{E}_g=\{(V_n,V_m)\mid V_n,V_m\in\mathcal{V}_g, ~n\neq m\}$,  $g=1,2,\dots, G$}; see Fig. \ref{ris_arch} (c).
 \item \emph{Tree-connected RIS}. Tree-connected RIS is defined as a class of BD-RIS architectures whose circuit topology forms a tree \cite{bondy2008graph}, i.e., $\mathcal{G}$ is connected and acyclic. 
Two examples of tree-connected RISs are tridiagonal RIS, where 
\begin{equation}\label{def:tri}
\bY=\left[\begin{matrix}
Y_{1,1}&Y_{1,2}&&\vspace{-0.1cm}\\
Y_{1,2}&Y_{2,2}&\ddots&\vspace{-0.0cm}\\
&\hspace{-0.5cm}\ddots&\hspace{-0.2cm}\ddots&\hspace{-0.35cm}Y_{N_I-1,N_I}\vspace{0.1cm}\\
&&\hspace{-0.5cm}Y_{N_I-1,N_I}&\hspace{-0.3cm}Y_{N_I,N_I}
\end{matrix} \right]
\end{equation}
and arrowhead RIS, where 
\begin{equation}\label{def:arrow}
\bY=\left[\begin{matrix}
Y_{1,1}&Y_{1,2}&\cdots&Y_{1,N_I}\\
Y_{1,2}&Y_{2,2}&&\vspace{-0.0cm}\\
\vdots&&\ddots&\vspace{0.1cm}\\
Y_{1,N_I}&&&\hspace{-0.3cm}Y_{N_I,N_I}
\end{matrix} \right].
\end{equation}The circuit topology of tridiagonal RIS and arrowhead RIS corresponds to a path graph and a  star graph \cite{bondy2008graph}, respectively; see Fig. \ref{ris_arch} (d) and Fig. \ref{ris_arch} (e), respectively.

\end{itemize}

The circuit topology of the reconfigurable impedance network determines both the performance and the complexity of the corresponding BD-RIS architecture.
The single-connected architecture minimizes circuit complexity but offers limited design flexibility.  Conversely, the fully-connected architecture provides the highest degree of freedom but comes with the highest circuit complexity. 
Between these two extremes lie various BD-RIS architectures, e.g., group- and tree-connected architectures,  which strike a balance between performance and complexity. In particular, for single-user MISO systems, tree-connected RIS has been shown to match the performance of fully-connected RIS with significantly reduced circuit complexity \cite{tree}.  
Building upon this result, this paper seeks to  \emph{identify new low-complexity BD-RIS architectures achieving the optimal performance in the general multiuser MIMO systems.}

\subsection{System Model}
Consider a BD-RIS-aided multiuser MIMO system consisting of  an $N_T$-antenna transmitter,  $K$ multi-antenna users, where the $k$-th user is equipped with $N_k$ received antennas,   and an $N_I$-element BD-RIS.  
The BD-RIS induces a scattering matrix $\bthe\in\C^{N_I\times N_I}$ to help shape the wireless channel. According to microwave network theory \cite{microwavebook}, the scattering matrix $\bthe$ is  related to the admittance matrix $\bY \in \C^{N_I \times N_I}$ as  follows
\begin{equation}\label{theta_b}
 \begin{aligned}
 \bthe&=(\mathbf{I}+Z_0\bY)^{-1}(\mathbf{I}-Z_0\bY)\\
 &=(\mathbf{I}+\mathsf{i}Z_0\bB)^{-1}(\mathbf{I}-\mathsf{i}Z_0\bB),
 \end{aligned}
 \end{equation}
where $Z_0$ denotes the reference impedance, and the second equality is due to \eqref{YandB}.
  Given $\bthe$,  the effective channel between the transmitter and the $k$-th user can be expressed as\footnote{We assume perfect matching, no mutual coupling, no structural scattering or specular reflection, and that the unilateral approximation
holds \cite{generalmodel}.} 
\begin{equation}\label{eff_channel}
\bH_{\text{eff},k}(\bthe)=\bH_{d,k}+\bH_{r,k}\bthe\bG,~~k=1,2,\dots, K,
\end{equation}
where $\bH_{d,k}\in\C^{N_k\times N_T}$,   $\bH_{r,k}\in\C^{N_k\times N_I}$, and $\bG\in\C^{N_I\times N_T}$  are the channels from the transmitter to the $k$-th user,  from the BD-RIS to the $k$-th user, and from the transmitter to the BD-RIS,  respectively. Here we use the notation $\bH_{\text{eff},k}(\bthe)$ to emphasize the dependence of the effective channel on $\bthe$. Let $\s_k\in\C^{d_k}$ be the transmit symbols for the $k$-th user with $\mathbb{E}\{\s_k\s_k^{H}\}=\mathbf{I}$, and let $\bW_k\in\C^{N_T\times d_k}$ be the beamforming matrix for the $k$-th user. Then the received signal at the $k$-th user is 
\begin{equation}\label{def:r}\mathbf{r}_k=\bH_{\text{eff},k}(\bthe)\sum_{k=1}^K\bW_k\s_k+\n_k,~~k=1,2,\dots, K,
\end{equation}
where $\n_k\sim\mathcal{CN}(0,\sigma^2\mathbf{I})$ is the additive white Gaussian noise. 
\subsection{Problem Formulation}
Given a specific BD-RIS architecture represented by a graph $\mathcal{G}=(\mathcal{V},\mathcal{E})$, the joint beamforming design and BD-RIS optimization problem can be modeled as follows: 
\begin{equation}\label{multiuser}
\begin{aligned}
P_{\mathcal{G}}=\max_{\bW,\bthe,\bB}~&F(\bW, \bH_{\text{eff}}(\bthe))\\
\text{s.t.}~~~ &\bthe=\left(\mathbf{I}+\mathsf{i}Z_0\mathbf{B}\right)^{-1}\left(\mathbf{I}-\mathsf{i}Z_0\mathbf{B}\right),~~\\
~~&\bB=\bB^{T},~~\bB\in\mathcal{B}_\mathcal{G},\\
&\bW\in\mathcal{W},
\end{aligned}
\end{equation}
where $
\bH_{\text{eff}}(\bthe)\hspace{-0.05cm}=\hspace{-0.05cm}[\bH_{\text{eff},1}(\bthe)^T,\dots,\bH_{\text{eff},K}(\bthe)^T]^{T}\in\C^{(\sum_{k=1}^K N_k)\times N_T}$ 
and $\bW=[\bW_1,\dots,\bW_K]\in\C^{N_T\times (\sum_{k=1}^Kd_k)}$ collect the effective channel and beamforming matrix of all users, respectively,  $F(\bW,\bH_{\text{eff}}(\bthe))$ is a general utility function,  $\bW\in\mathcal{W}$ accounts for the constraint at the transmitter side, and 
$$\mathcal{B}_\mathcal{G}:=\{\mathbf{B}\mid B_{n,m}=0,~\forall\,n\neq m\text{ with } (V_n,V_m)\notin\mathcal{E}\}
$$ incorporates the BD-RIS architecture into variable $\bB$. 
\begin{remark}\label{example}
The function $F(\bW, \bH_{\text{\normalfont{eff}}}(\bthe))$ in \eqref{multiuser} encompasses many, if not all,  utility functions commonly considered in the literature. Some examples are as follows. 
\begin{itemize}  \item[(a)] Sum channel gain maximization \cite{qstem}: 
$F=\|\bH_{\text{\normalfont{eff}}}(\bthe)\|_F^2;$ 
\item[(b)] Sum-rate maximization \cite{PDD,wu}: 
$$F(\bW, \bH_{\text{\normalfont{eff}}}(\bthe))=\hspace{-0.1cm}\sum_{k=1}^KR_k(\bW,\bH_{\text{\normalfont{eff}},k}(\bthe)),$$
where the rate of the $k$-th user is given by 
 {$$
  \begin{aligned}
R_k(\bW,\hspace{-0.02cm}\bH_{\text{\normalfont{eff}},k}(\bthe))\hspace{-0.05cm}=\hspace{-0.05cm}\log\hspace{-0.02cm}\det\hspace{-0.06cm} \bigg( \mathbf{I} \hspace{-0.05cm}+\hspace{-0.05cm}  \mathbf{H}_{\text{\normalfont{eff}},k}(\bthe) \mathbf{W}_k \mathbf{W}_k^H \mathbf{H}_{\text{\normalfont{eff}},k}(\bthe)^H \nonumber \\
\times \bigg( \sum_{j \neq k} \mathbf{H}_{\text{\normalfont{eff}},k}(\bthe) \mathbf{W}_j \mathbf{W}_j^H \mathbf{H}_{\text{\normalfont{eff}},k}(\bthe)^H + \sigma^2 \mathbf{I} \bigg)^{-1} \bigg);
\end{aligned}
$$}
 \item[(c)] Energy efficiency maximization \cite{PDD}: 
$$F(\bW, \bH_{\text{\normalfont{eff}}}(\bthe))=\frac{\sum_{k=1}^KR_k(\bW,\bH_{\text{\normalfont{eff}},k}(\bthe))}{\eta\|\bW\|_F^2+P_c},$$
where $\eta$ and $P_c$ are the power amplifier efficiency and the energy expenditure required by the device operating, respectively;
 \item[(d)] Transmit power minimization under users' QoS requirements \cite{PDD,wu}:
 $$
 \begin{aligned}F(\bW, \bH_{\text{\normalfont{eff}}}(\bthe))=&-\|\bW\|_F^2-\sum_{k=1}^K\mathbb{I}_{\{R_k(\bW,\H_{\text{\normalfont{eff}},k}(\bthe))\geq \gamma_k\}}(\bW,\bthe),
 \end{aligned} $$
 where $\gamma_k$ is the rate threshold of the $k$-th user, and  $\mathbb{I}_{\mathcal{X}}(x)$ refers to the indicator function
of set $\mathcal{X}.$
 \end{itemize}
 On the other hand, $\bW\in\mathcal{W}$ can represent different types of power constraints at the transmitter side. Some examples are as follows.
\begin{itemize}
\item[(a)] No power budget: $\mathcal{W}=\C^{N\times K};$
\item [(b)]Total transmit power constraint: $\mathcal{W}=\{\|\bW\|_F^2\leq P_T\}$, where $P_T$ is the maximum transmit power at the transmitter;
\item[(c)]  Per-antenna transmit power constraint: $\mathcal{W}=\left\{\sum_{k=1}^K\|\bW_k(n,:)\|_2^2\leq \frac{P_T}{N_T},~n=1,2,\dots, N_T\right\};$
\item[(d)] Constant-envelope power constraint:  $\mathcal{W}=\left\{\sum_{k=1}^K\|\bW_k(n,:)\|_2^2= \frac{P_T}{N_T},~n=1,2,\dots, N_T\right\}.$\end{itemize}
\end{remark}
In this paper, we aim to identify new BD-RIS architectures with low circuit complexity that  achieve the same performance as fully-connected RIS. Since the BD-RIS architecture is generally determined offline, the above property needs to  be satisfied for any channel realization.   Mathematically, the problem translates to seeking a graph $\mathcal{G}$ such that for any problem parameter inputs (i.e., $\bG, \{\bH_{d,k}\}_{k=1}^K,\{\bH_{r,k}\}_{k=1}^K$),  the following result holds:
$$P_{\mathcal{G}}=P_{\text{fully}},$$
where $P_{\text{fully}}$ denotes the optimum achieved by fully-connected RISs, which is, the optimal value of \eqref{multiuser} with $\mathcal{B}_{\mathcal{G}}=\R^{N_I\times N_I}$.  It is worth noting that by focusing on the general model in \eqref{multiuser}, the optimality of our proposed architectures would hold regardless of the performance metric. 

\section{Main Result: the Optimal Architecture}\label{sec:3}
In this section, we discuss our main result. The result is summarized in Section \ref{sec:31}, which gives a class of BD-RIS architectures  with low circuit complexity that match the performance of fully-connected RIS. In particular, we introduce two novel architectures\,---\,band-connected RIS and stem-connected RIS\,---\,that belong to the optimal architecture class in Section \ref{sec:32}. Finally, we provide some further discussions on the main result in Section \ref{sec:33}. 
\vspace{-0.3cm}

\subsection{Main Results}\label{sec:31}

Before presenting our main result, we give the definition of the adjacency matrix, which describes the connectivity of a graph  in an algebraic way. 
\begin{definition}[Adjacency Matrix \cite{bondy2008graph}]
Given a graph $\mathcal{G}$ with $N$ vertices, its adjacency matrix is defined as an $N\times N$ matrix $\bA_\mathcal{G}$, where $\bA_{\mathcal{G}}(n,m)$ represents the number of edges connecting vertices $V_n$ and $V_m$. 
\end{definition}
\hspace{-0.35cm}Note that for the reciprocal BD-RIS considered in this paper, its circuit topology forms  an undirected simple graph, i.e.,  all the edges are bidirectional and at most one edge exists between any pair of vertices. In this case, the adjacency matrix is a symmetric 0\,--\,1 matrix with zeros on its diagonal. 
  It is convenient to characterize the topological connectivity of the BD-RIS circuit using the adjacency matrix. 
Specifically, the susceptance matrix $\bB$ and the adjacency matrix of its corresponding graph are related as follows: 
\begin{equation}\label{def:adj}
\bA_\mathcal{G}(n,m)=\left\{
\begin{aligned}
1,~~~&\text{if }B_{n,m}\neq 0, ~n\neq m;\\
0,~~~&\text{otherwise}, 
\end{aligned}\right.
\end{equation}
i.e.,  $\bA_\mathcal{G}(n,m)=1$ if there is an interconnection between the $n$-th and $m$-th ports of the reconfigurable impedance network, and $\bA_{\mathcal{G}}(n,m)=0$ otherwise.

With the definition above, we are now ready to present our main result.  The main result, stated below,  provides a general condition on the adjacency matrix under which the corresponding BD-RIS architecture achieves the same optimal performance as fully-connected RIS. 
\begin{theorem}\label{theorem1}
 The BD-RIS architecture induced by $\mathcal{G}=(\mathcal{V}, \mathcal{E})$ achieves the same optimal performance as fully-connected RIS, i.e., $P_{\mathcal{G}}=P_{\text{\normalfont fully}}$, {\color{black}almost surely}\footnote{{\color{black}Roughly speaking, this means that $P_{\mathcal{G}}=P_{\text{\normalfont fully}}$ for almost all channel realizations, with exceptions being extremely rare and negligible.}} if there exists a permutation matrix $\bP\in\R^{N_I\times N_I}$ such that the adjacency matrix of  $\mathcal{G}$ can be written as  
\begin{equation}\label{adjacencymatrix}
\bA_{\mathcal{G}}=\bP\bar{\bA}_{\mathcal{G}}\bP^T,
\end{equation} 
where $\bar{\bA}_{\mathcal{G}}$ satisfies the following condition
\begin{equation}\label{graph_condition}
\sum_{m=n+1}^{N_I}\hspace{-0.1cm}\bar{\bA}_{\mathcal{G}}({n,m})=\min\{2L-1, N_I-n\},~1\leq n\leq N_I-1\\
\end{equation}
{\color{black}with 
\begin{equation}\label{def:L}
L=\min\left\{D,\frac{N_I}{2}\right\}~\text{  and  }~D=\min\left\{\sum_{k=1}^KN_k,N_T\right\}.
\end{equation}}
\end{theorem}\vspace{-0.9cm}
\begin{proof}
See Section \ref{sec:4}.
\end{proof}

Note that given two graphs $\mathcal{G}_1$ and $\mathcal{G}_2$, their adjacency matrices satisfy $\mathbf{A}_{\mathcal{G}_1}=\bP\mathbf{A}_{\mathcal{G}_2}\bP^T$ for some permutation matrix $\bP$ if and only if these two graphs are isomorphic{\color{black}\cite{bondy2008graph}.} Consequently, the graphs defined in Theorem \ref{theorem1} can also be described as those isomorphic to  graphs with adjacency matrices satisfying condition \eqref{graph_condition}. To give a more intuitive illustration, we present the structure of adjacency matrices satisfying condition \eqref{graph_condition} in Fig. \ref{graph_condition_matrix}. Focusing on the upper tridiagonal part, the first  $N_I-2L$ rows of the matrix are sparse, each containing $2L-1$ non-zero elements, whereas the remaining rows are fully populated. The corresponding graph has  each of its first $N_I-2L$ vertices  connected to arbitrary $2L-1$ of its subsequent vertices, {\color{black}where the subsequent vertices of $V_i$ refer to the vertices with indices larger than $i$, i.e., $V_{i+1},V_{i+2},\dots, V_{N_I}$,} 
and the last $2L$ vertices form a complete graph. {\color{black}In particular, when $L=\frac{N_I}{2}$, the corresponding architecture is fully-connected RIS.} 
\begin{figure}
\includegraphics[scale=0.23]{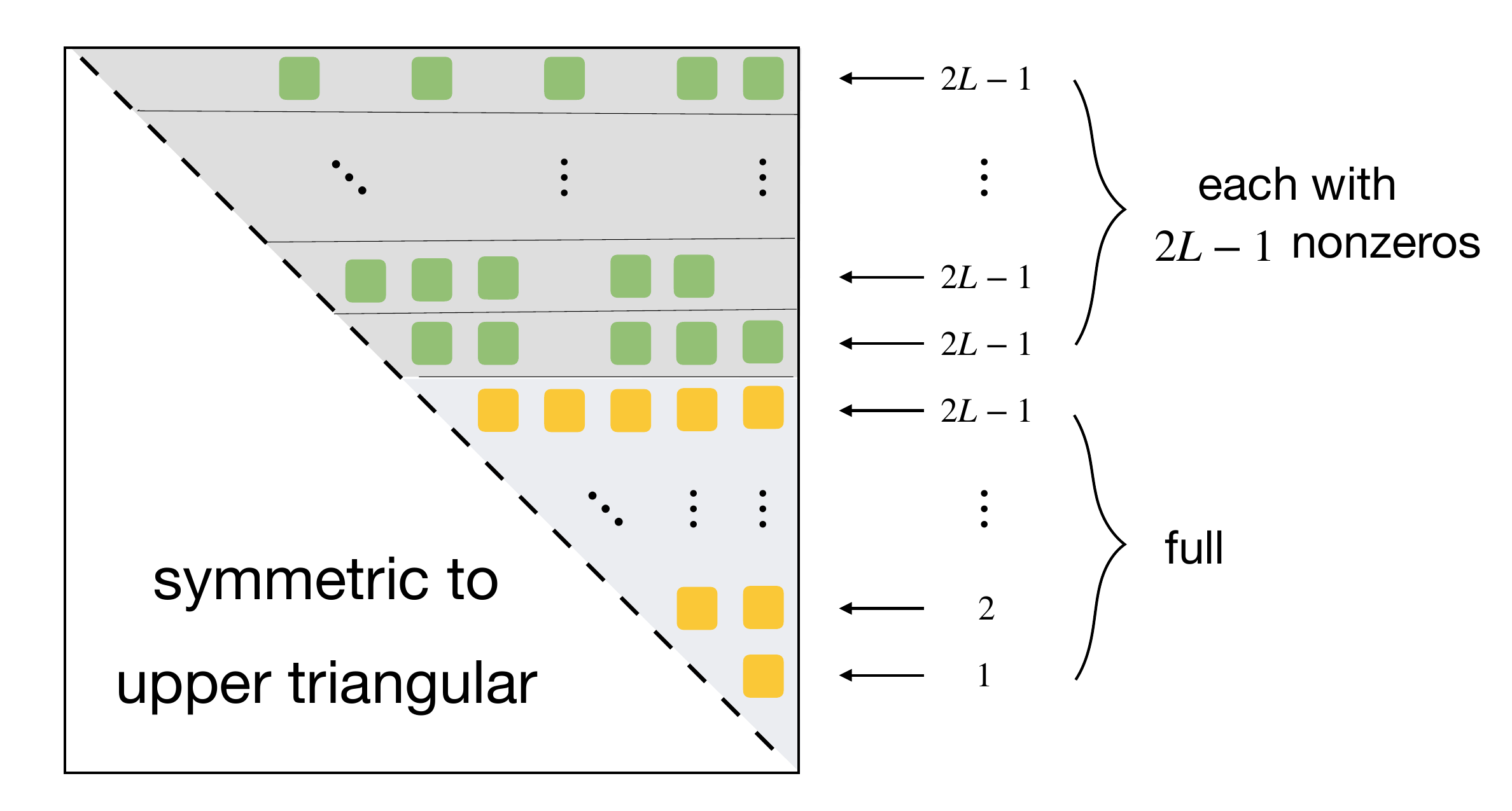}
\centering
\caption{An illustration of adjacency matrices satisfying condition \eqref{graph_condition}, where the green and yellow squares indicate the positions of the nonzero elements. {\color{black} For clarity, only the upper triangular part of the matrix is shown; the lower triangular part is determined by the symmetry of the matrix.}}
\label{graph_condition_matrix}
\end{figure}

\begin{remark}[Circuit complexity] \label{remark:complexity}
{\color{black}According to \eqref{def:adj}, the total number of admittances in a BD-RIS architecture can be calculated as 
$$\text{Num}=N_I+\frac{1}{2}\sum_{n=1}^{N_I}\sum_{m=1}^{N_I}\mathbf{A}_{\mathcal{G}}(n,m),$$
where the first term accounts for the self-admittances of the $N_I$ ports, and the second term records  the number of admittances  for interconnection between ports.  For the BD-RIS architectures specified in Theorem \ref{theorem1}, we further have 
 $$\begin{aligned}
 \frac{1}{2}\sum_{n=1}^{N_I}\sum_{m=1}^{N_I}\mathbf{A}_{\mathcal{G}}(n,m)&\overset{(a)}{=}\frac{1}{2}\sum_{n=1}^{N_I}\sum_{m=1}^{N_I}\mathbf{\bar{A}}_{\mathcal{G}}(n,m)\\
&\overset{(b)}{=}\sum_{n=1}^{N_I-1}\sum_{m=n+1}^{N_I}\mathbf{\bar{A}}_{\mathcal{G}}(n,m)\\
&\overset{(c)}{=}\sum_{n=1}^{N_I-1}\min\{2L-1,N_I-n\}=(N_I-L)(2L-1),
\end{aligned}$$
 where (a) uses the fact that $\bP$ is a permutation matrix, (b) is due to the symmetry of the adjacency matrix, and the (c) uses the condition in \eqref{graph_condition}. Hence, the total number of admittances for the BD-RIS architectures in Theorem \ref{theorem1} is $$\text{Num}=N_I+(N_I-L)(2L-1)=L(2N_I-2L+1),$$} \hspace{-0.1cm}which is in the order of $N_IL$.
  In comparison, fully-connected RIS requires $N_I(N_I+1)/2$ admittances, which scales as $N_I^2$. Given that the number of transmit antennas $N_T$ and  the total number of received antennas across all users $\sum_{k=1}^K N_k$ are typically much smaller than the number of RIS elements $N_I$ in practice, the proposed scheme exhibits significantly lower circuit complexity than fully-connected RISs.
\end{remark}

\begin{remark}[Relationship between the circuit complexity and the DoF of the multiuser MIMO channel]
It is interesting to note that $D$ is exactly the DoF, also known as the multiplexing gain, of the considered multiuser MIMO channel, which characterizes the maximum number of interference-free data streams that can be transmitted in parallel \cite{DoF,DoF2}.  This illustrates that, in order to achieve full  flexibility as provided by fully-connected RIS, it suffices that the circuit complexity of BD-RIS scales as $N_I$ times the DoF of the multiuser MIMO channel (as $N_I\gg D$ in practice).
\end{remark}
\begin{remark}[Connection with the result in \cite{tree}]\label{remark:tree}
For single-user MISO systems, i.e., $K=N_k=1$, tree-connected RIS has been proved to achieve the same performance as fully-connected RIS in terms of maximizing the channel gain.  Interestingly, we can show that when $K=1$ and $N_k=1$,   the conditions specified  by \eqref{adjacencymatrix} and \eqref{graph_condition} in Theorem \ref{theorem1} are satisfied if and only if  $\mathcal{G}$ is a tree, which aligns with the result in \cite{tree}. We relegate the proof of this claim to Appendix \ref{app:tree}.
\end{remark}

\begin{remark}[Impact of the number of data streams $\{d_k\}_{k=1}^K$ on the optimal architecture]
Theorem \ref{theorem1} holds for any given $\{d_k\}_{k=1}^K$. In practice, $d_k$ satisfies $d_k\leq N_k$ for $k=1,2,\dots, K$, and achieving the capacity of the multiuser MIMO channel may require setting $d_k=N_k, ~k=1,2,\dots, K$.  It is worth mentioning that, when $\sum_{k=1}^Kd_k<D$, i.e., when the number of transmitted streams is smaller than the DoF, it is possible to achieve the performance of fully-connected RIS 
 with lower circuit complexity than that specified by Theorem \ref{theorem1}. Specifically, if the utility  function takes  the following form\footnote{The motivation for considering the utility function in \eqref{def:F} is that $\bH_{\text{\normalfont eff}}(\bthe)$ and $\bW$ affect the received signals \eqref{def:r} through their product  $\bH_{\text{\normalfont eff}}(\bthe)\bW$, which we refer to as the precoded effective channel matrix. In particular, all utility functions given in Remark \ref{example} belong to the function class in \eqref{def:F}, except for the sum channel gain, which is excluded as it is a function of the effective channel $\bH_{\text{\normalfont eff}}(\bthe)$.
} 
\begin{equation}\label{def:F}
F(\bW,\bH_{\text{\normalfont eff}}(\bthe))=G(\bW,\bH_{\text{\normalfont eff}}(\bthe)\bW),
\end{equation}
{\color{black}where $G: \R^{(N_T+\sum_{k=1}^K N_k)\times \sum_{k=1}^K d_k}\rightarrow \R$ is an arbitrary function.} Theorem \ref{theorem1} holds true with $D$ replaced by 
$\tilde{D}=\sum_{k=1}^Kd_k.$
Roughly speaking, this is because $\tilde{D}$ is the rank of the  precoded effective channel matrix $\bH_{\text{\normalfont eff}}(\bthe)\bW$, which represents the overall end-to-end channel experienced by the transmitted symbols.  The proof builds on a similar idea as Theorem \ref{theorem1}. We omit the details to avoid redundancy.

As a direct consequence of Remark \ref{remark:tree} and the above discussions, tree-connected RIS is optimal for a point-to-point MIMO system with a single data stream. To the best of our knowledge, this result has not been established in the literature.

\end{remark}

\subsection{Two Examples of the Optimal Architectures}\label{sec:32}
To provide further insights into the optimal architecture class given in Theorem \ref{theorem1}, we next present two specific examples: band-connected RIS and stem-connected RIS, which extend the concepts of tridiagonal RIS and arrowhead RIS discussed in Section \ref{existing_arch}, respectively.
\subsubsection{Band-connected RIS}
Band-connected RIS is defined as the BD-RIS architecture whose susceptance matrix $\bB$ is a band matrix given as follows:
\begin{equation}\label{band}
\bB=\left[\hspace{0.1cm}\begin{matrix}
B_{1,1}&\hspace{-0.2cm}\cdots&\hspace{-0.6cm}B_{1,q+1}&\hspace{-0.2cm}&\hspace{-0.2cm}&\hspace{-0.5cm}\\
\vdots&\ddots&\hspace{-0.8cm}&\hspace{-1cm}\ddots&\hspace{-1cm}&\\
B_{1,q+1}&\hspace{-0.1cm}&\ddots\hspace{0.8cm}&&\hspace{-0.6cm}B_{N_I-q,N_I}\\
&\hspace{-0.1cm}\hspace{-0.5cm}\ddots&\hspace{1cm}\ddots&\hspace{-0.2cm}&\hspace{-0.6cm}\vdots\\
&\hspace{-0.2cm}&\hspace{-1cm}B_{N_I-q,N_I}&\hspace{-1cm}\cdots\hspace{-0.2cm}&\hspace{-0.3cm}B_{N_I,N_I}
\end{matrix}\hspace{-0.2cm}\right],
\end{equation}
i.e., $B_{n,m}=0$ if $|m-n|> q$, where $q$ is referred to as the band width. When  $q=1$, the matrix in \eqref{band} reduces to the tridiagonal matrix in \eqref{def:tri}, and the corresponding BD-RIS architecture is a tridiagonal RIS. The circuit topology of band-connected RIS forms a $q$-step path graph, which is a generalization of the path graph given in Fig. \ref{ris_arch} (d); see the definition below. 
\begin{figure}
\includegraphics[width=0.7 \columnwidth]{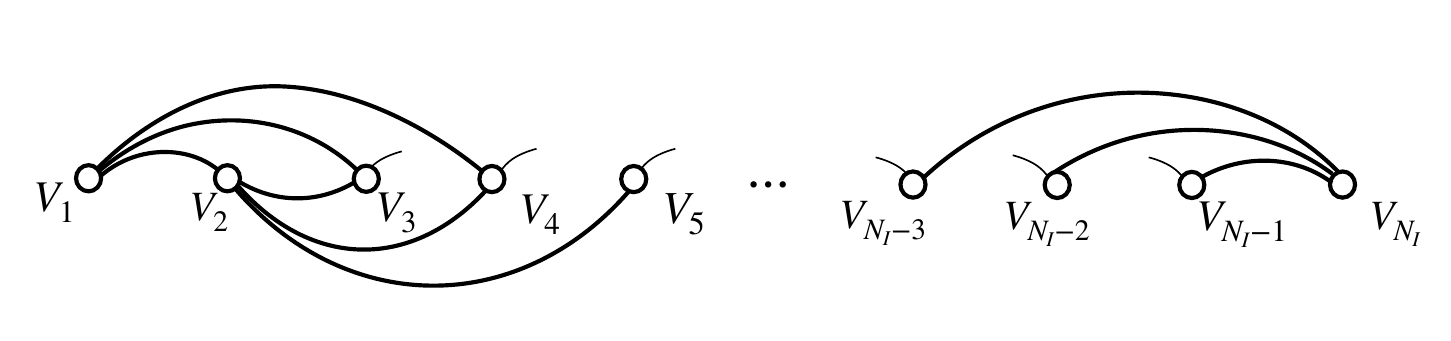}
\centering
\vspace{-0.5cm}
\caption{An illustration of the $q$-step path graph, where $q=3$.}
\label{qstep}
\end{figure}
\begin{definition}[$q$-step path graph] \label{def:npath}
A graph $\mathcal{G}=(\mathcal{V},\mathcal{E})$ is called a $q$-step path graph if each vertex  is connected to its q-nearest neighbors. 
Formally, let  $\mathcal{V}=\{{V}_1,{V}_2,\dots,{V}_{N_I}\}$, then the edge set is defined as  
$$\mathcal{E}=\left\{(V_n,V_m)\mid |m-n|\leq q,~~1\leq n,m\leq N_I,~m\neq n\right\}.$$ 
In particular, when $q=1$, the $q$-step path graph reduces to the classical path graph.
\end{definition} 
\hspace{-0.35cm}For illustration, a $q$-step  path graph is shown in Fig. \ref{qstep}.

 As a direct corollary of Theorem \ref{theorem1}, the following result demonstrates the optimality of band-connected RIS. 

\begin{corollary}
A band-connected RIS with band width $q=2L-1$ achieves the same performance as the fully-connected RIS. 
\end{corollary}
\begin{proof}
{\color{black} It suffices to show that the adjacency matrix of a  $(2L-1)$-step path graph satisfies conditions  \eqref{adjacencymatrix} and \eqref{graph_condition}  specified in Theorem \ref{theorem1}. Specifically, }
the adjacency matrix of a  $(2L-1)$-step path graph is given by\footnote{Here we only give the upper triangular part of $\bA_{\mathcal{G}}$, the lower triangular part can be determined immediately due to the symmetry of $\bA_{\mathcal{G}}$.}
$$
\begin{aligned}
&\bA_{\mathcal{G}}(n,m)=\left\{
\begin{aligned}
&1,~~~~\text{if }n+1\leq m\leq \min\{n+2L,N_I\};\\
&0,~~~~\text{otherwise},
\end{aligned}\right.~~~\forall~1\leq n<m\leq N_I.
\end{aligned}$$
{\color{black}Clearly, it satisfies} \eqref{adjacencymatrix} and \eqref{graph_condition} {\color{black}by setting} $\mathbf{P}=\mathbf{I}$  {\color{black}in \eqref{adjacencymatrix}.}
\end{proof}
\textcolor{black}{We remark that the band-connected RIS defined in \eqref{band} can be generalized such that each vertex is connected to any $2L-1$ of its subsequent vertices, not necessarily adjacent ones  (if a vertex has fewer than $2L-1$ subsequent vertices, it connects to all of them). This generalization actually corresponds to the configuration depicted in Fig. \ref{graph_condition_matrix}.}

\subsubsection{Stem-connected RIS \cite{qstem}} Stem-connected RIS is defined as the BD-RIS architecture whose susceptance matrix $\bB$  has the following structure: 

 {\small\begin{equation}\label{center}
 \mathbf{B}\hspace{-0.05cm}=\hspace{-0.05cm}
 \left[\hspace{-0.05cm}\begin{array}{c:c}\begin{matrix}
B_{1,1}\hspace{0.1cm}&\cdots&\hspace{0.1cm}B_{1,q}\\
 \hspace{-0.05cm}\vdots&\hspace{-0.1cm}\ddots\hspace{-0.1cm}&\vdots\\
B_{1,q}\hspace{0.1cm}&\hspace{-0.1cm}\cdots\hspace{-0.1cm} &\hspace{0.1cm}B_{q,q}\\
\end{matrix}\hspace{-0.1cm}& \hspace{-0.1cm}
 \begin{matrix}
B_{1,q+1}\hspace{0.05cm}&\cdots\hspace{-0.15cm}&B_{1,N_I}\\
\vdots&\hspace{-0.15cm}\ddots\hspace{-0.15cm}&\vdots\\
B_{q,q+1}\hspace{0.05cm}&\hspace{-0.15cm}\cdots\hspace{-0.15cm}&B_{q,N_I}\end{matrix}\vspace{0.07cm}\\
\hdashline
  \begin{matrix}
 \hspace{-0.05cm} B_{1,q+1}\hspace{-0.15cm}&\cdots&\hspace{-0.05cm}B_{q,q+1}\\
\hspace{-0.05cm}  \vdots&\ddots&\vdots\\
\hspace{-0.05cm}B_{1,N_I}&\cdots&B_{q,N_I}\end{matrix}\hspace{-0.1cm}&\hspace{-0.1cm}\begin{matrix}
B_{q+1,q+1}\hspace{-0.05cm}&&\\&\hspace{-0.1cm}\ddots&\\&&\hspace{-0.2cm}B_{N_I,N_I}\end{matrix}\hspace{-0.05cm}\end{array}\hspace{-0.1cm}\right],
\end{equation}}
\hspace{-0.15cm}i.e.,  $B_{n,m}=0$ for $ n,m >  q$ with $n\neq m$, where we call $q$ the stem width. When $q=1$, the matrix in \eqref{center} reduces to the arrowhead matrix in \eqref{def:arrow}, and the corresponding architecture is the arrowhead RIS. The circuit topology of a stem-connected RIS with stem width $q$ corresponds to a $q$-center graph, which generalizes the star graph shown in Fig. 1(e); see the definition below.
\begin{definition}[$q$-center graph] \label{def:npath}
A graph $\mathcal{G}=(\mathcal{V},\mathcal{E})$ is called a $q$-center graph if it contains $q$ central vertices, each connected to all other vertices, while the remaining vertices are only connected to the central vertices. Specifically, denote the central vertices as 
$\mathcal{V}_c\subseteq \mathcal{V}=\{V_1,V_2,\dots,V_{N_I}\}$, where $|\mathcal{V}_c|=q$. The edge set is then given by  
$$\mathcal{E}=\{(V_n,V_m)\mid V_n\in\mathcal{V}_c, ~\mathcal{V}_m\in\mathcal{V},~1\leq n,m\leq N_I,~ n\neq m\}.$$
When $q=1$, the $q$-center graph reduces to  a star graph.
\end{definition}
\begin{figure}
\includegraphics[width=0.6 \columnwidth]{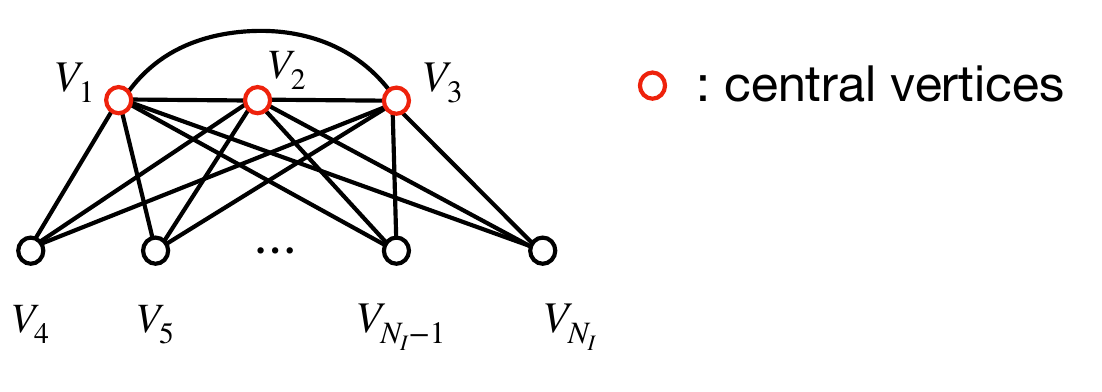}
\centering
\vspace{-0.5cm}
\caption{An illustration of the $q$-center graph, where $q=3$ {\color{black} and $\mathcal{V}_c=\{V_1,V_2,V_3\}$.}}
\label{qcenter}
\end{figure}
 \hspace{-0.35cm}An illustrative example of $q$-center graph is given in Fig. \ref{qcenter}.  {\color{black}For the stem-connected architecture defined in \eqref{center}, the central vertices are the first $q$ vertices, i.e., $\mathcal{V}_c=\{V_1,V_2,\dots,V_q\}$. The BD-RIS architecture corresponding to other set of central vertices can be defined similarly and is also optimal (as long as $q=2L-1$) by similar arguments as in Corollary \ref{corollary:stem}.}
 

\begin{corollary}\label{corollary:stem}
A stem-connected RIS with stem width $q=2L-1$ achieves the same performance as the fully-connected RIS. 
\end{corollary}
\begin{proof}

Note that for a stem-connected RIS with stem width $q$, the corresponding adjacency matrix  is given by
 \begin{equation*}\label{qstem}
{\bA}_{\mathcal{G}}(n,m)=\left\{
\begin{aligned}
&1,~~~~\text{if }m\leq q \text{ or }n\leq q,~m\neq n;\\
&0,~~~~\text{otherwise}.
\end{aligned}\right.
\end{equation*}
When the stem width $q=2L-1$, the above matrix satisfies \eqref{adjacencymatrix} and \eqref{graph_condition} in Theorem \ref{theorem1} with  $\bP=[\mathbf{e}_{N_I},\mathbf{e}_{N_I-1},\dots,\mathbf{e}_1]$.
\end{proof}

We remark that the stem-connected RIS was first introduced in a recent work \cite{qstem}, where it  was numerically shown to achieve the same performance as the fully-connected RIS in terms of maximizing the sum channel gain of multiuser MISO systems. However, the optimality of stem-connected RIS has not been theoretically established in \cite{qstem}.

To conclude this subsection, we give a summary of existing BD-RIS architectures and their corresponding circuit complexity in Table \ref{summary}.
\begin{table*}
\fontsize{7}{7}\selectfont
\caption{A summary of existing BD-RIS architectures and the corresponding circuit complexity.}
\begin{tabular}{c|c|c|c}
\hline
\multirow{2}{*}{Architecture}&\multirow{2}{*}{Number of admittances}&\multirow{2}{*}{Property}&\multirow{2}{*}{Examples}\\
&&&\\
\hline
\multirow{2}{*}{single-connected RIS}&\multirow{2}{*}{$N_I$}&{least design flexibility,}&\multirow{2}{*}{$\backslash$}\\
&&least circuit complexity&\\
\hline
\multirow{2}{*}{group-connected RIS}&\multirow{2}{*}{$\frac{N_I}{2}\left(\frac{N_I}{G}+1\right)$}&balance between single-connected RIS&$G=N_I$: single-connected RIS\\
&&and fully-connected RIS&$G=1$: fully-connected RIS\\
\hline
\multirow{2}{*}{fully-connected RIS} &\multirow{2}{*}{$\frac{N_I(N_I+1)}{2}$}&{highest design flexibility,}&\multirow{2}{*}{$\backslash$}\\
&&highest circuit complexity&\\
\hline
\multirow{2}{*}{tree-connected RIS} &\multirow{2}{*}{$2N_I-1$}&{optimal performance with least circuit}&\multirow{2}{*}{tridiagonal RIS and arrowhead RIS}\\
&&  complexity in single-user MISO systems&\\
\hline
{proposed architectures}&{$L(2N_I\hspace{-0.05cm}-\hspace{-0.05cm}2L\hspace{-0.05cm}+\hspace{-0.05cm}1)$, where $L=\min\{D, \frac{N_I}{2}\}$}&optimal performance&{band-connected RIS and stem-connected RIS}\\
in Theorem \ref{theorem1}&and $D=\min\{\sum_{k=1}^K\hspace{-0.05cm} N_k, N_T\}$&in multiuser MIMO systems&with band and stem width $q=2L-1$\\
\hline
\end{tabular}
\label{summary}
\centering
\end{table*}

\subsection{Discussions on Theorem \ref{theorem1}}\label{sec:33}
In this subsection, we give some further discussions on Theorem \ref{theorem1}. 

First, we discuss the practical applicability of Theorem \ref{theorem1}. According to its statement, the proposed architectures can match the performance of fully-connected RIS  at the optimum, i.e., when problem \eqref{multiuser} is solved to global optimality. However, globally solving \eqref{multiuser} is infeasible  in practice  due to its highly non-convex nature and large dimensionality. 
Nevertheless, we can actually obtain a stronger result than Theorem \ref{theorem1}. Let $\mathcal{G}$ be a graph satisfying the conditions in Theorem \ref{theorem1}, then the following result holds:
\begin{equation}\label{generalresult}
\begin{aligned}
\mathcal{H}_{\text{fully}}=\mathcal{H}_{\mathcal{G}}\cup\mathcal{N},
\end{aligned}
\end{equation}
where 
{
\begin{equation}\label{Hfully}
\begin{aligned}\mathcal{H}_{\text{fully}}=&\{\bH_{\text{eff}}(\bthe)\mid \bthe=\left(\mathbf{I}+\mathsf{i}Z_0\mathbf{B}\right)^{-1}\hspace{-0.05cm}\left(\mathbf{I}-\mathsf{i}Z_0\mathbf{B}\right),
\bB=\bB^T \}
\end{aligned}
\end{equation}}
and
{\begin{equation}\label{Hproposed}
\begin{aligned}
\mathcal{H}_{\mathcal{G}}=\{\bH_{\text{eff}}(&\bthe)\hspace{-0.05cm}\mid \hspace{-0.05cm}\bthe=\left(\mathbf{I}+\mathsf{i}Z_0\mathbf{B}\right)^{-1}\hspace{-0.05cm}\left(\mathbf{I}-\mathsf{i}Z_0\mathbf{B}\right),
\,\bB=\bB^T,~\bB\in\mathcal{B}_\mathcal{G}\}
\end{aligned}\end{equation}}
\hspace{-0.15cm}are the effective channel sets achieved by fully-connected RIS and the architecture induced by $\mathcal{G}$, respectively, and $\mathcal{N}$ is a low-dimensional subset of  {\color{black}$\mathcal{H}_{\text{fully}}$}; see Theorem \ref{theorem2} in Section \ref{sec:4} for a rigorous statement. 
This means that the proposed architecture can achieve (almost) the same flexibility in channel shaping as the fully-connected RIS\footnote{The set $\mathcal{N}$ can be ignored as its dimension is lower than that of  {\color{black}$\mathcal{H}_{\text{fully}}$}. Intuitively, this means that its volume is negligible compared to that of \hspace{-0.03cm}{\color{black}$\mathcal{H}_{\text{fully}}$}.}.  Therefore, given a solution corresponding to a fully-connected RIS (which might be obtained via some numerical optimization approach),  we can, with probability one, construct a solution $\bB\in\mathcal{B}_{\mathcal{G}}$ for the proposed architecture that achieves the same performance. 
This statement is stronger and more practically  applicable than the one presented in Theorem \ref{theorem1}.  As another key implication of \eqref{generalresult}, we note that the optimality of the proposed architectures extends beyond the performance metric in \eqref{multiuser} for linear precoding --- they also achieve the same information-theoretic limits as fully-connected RIS. Specifically, the capacity region \cite{capacity} attained by a BD-RIS aided MU-MIMO channel is given  as follows:  
 $$\mathcal{C}=\text{Conv}\left\{\bigcup_{\bH\in\mathcal{H}}\bigcup_{\pi}\bigcup_{\sum_{k=1}^K\hspace{-0.08cm}\text{Tr}(\bQ_k)=P_T\atop \bQ_k\geq 0, \,\forall k}\{(R_1,R_2,\dots,R_K)\}\right\},$$
where $\mathcal{H}$ is the effective channel set achieved by a given BD-RIS architecture, $\pi$ is a permutation of $\{1,2,\dots, K\}$ defining the encoding order, $\bQ_k$ is the transmit covariance matrix for the $k$-th user,  $R_k$ is the achievable rate of the $k$-th user, with $R_{\pi(q)}$ given by
$$R_{\pi(q)}=\log\det\left(\mathbf{I}+\bH_{\pi(q)}\bQ_{\pi(q)}\bH_{\pi(q)}^H\left(\sigma^2\mathbf{I}+\bH_{\pi(q)}\sum_{p>q}\bQ_{\pi(p)}\bH_{\pi(q)}\right)^{-1}\right),~~q=1,2,\dots, K,$$ and $\text{Conv}\{\mathcal{X}\}$ denotes the convex hull of set $\mathcal{X}$. 
According to \eqref{generalresult} and due to the fact that $\mathcal{N}$ is a low-dimensional subspace of $\mathcal{H}_{\mathcal{G}}$, we can conclude that the capacity region attained by the proposed architecture and the fully-connected RIS are  (almost) the same, differing only in a set of measure zero. 

Next, we provide insights into why it is possible to achieve the performance of fully-connected RIS with significantly fewer connections when the number of RIS elements $N_I$ is large. The following discussion is heuristic, and a rigorous proof can be found in Section \ref{sec:4}.  We take $D=\sum_{k=1}^K N_k$ as an example to give some intuitions. 
 In this case, we introduce a sequence of new variables: 
\begin{equation}\label{def:uk}
\bU_k=(\bH_{r,k}\bthe)^{H}\in\C^{N_I\times N_k},~  k=1,2,\dots, K,
\end{equation}
which  allows us to express the effective channel in \eqref{eff_channel}  as 
\begin{equation*}\label{}
\bH_{\text{eff},k}=\bH_{d,k}+\bU_k^{H}\bG,~~k=1,2,\dots, K.
\end{equation*}
Clearly, the impact of BD-RIS on the effective channel is fully characterized by the new variables $\{\bU_k\}_{ k=1}^ K$. If  full flexibility in  $\{\bU_k\}_{ k=1}^ K$  (as provided by fully-connected RIS) is achieved, the  full potential of BD-RIS is realized. 
Due to the relationship between $\bthe$ and $\bB$, $\bU_k=(\bH_{r,k}\bthe)^H$ can be expressed   as a linear system over $\bB$ and $\bU_k$ as follows:
\begin{equation}\label{system}
\left(\mathbf{I}-\mathsf{i}Z_0\mathbf{B}\right)\bU_k=\left(\mathbf{I}+\mathsf{i}Z_0\mathbf{B}\right)\bH_{r,k}^H,~k=1,2,\dots K.\end{equation}
Hence, $\{\bU_k\}_{ k=1}^ K$ are achievable by the BD-RIS architecture induced by $\mathcal{G}$ if and only if the above linear system has a solution $\bB\in\mathcal{B}_{\mathcal{G}}$ with $\bB=\bB^T$.  Since the total number of linear equations  in \eqref{system} is $N_ID$, we only require the degree of freedom (i.e., the nonzero elements) in $\bB$ to also be in the order of $N_ID$ to guarantee the existence of a solution. 
This explains why it is possible to achieve full flexibility of BD-RIS (as provided by fully-connected RIS)  with circuit complexity in the order of $N_ID$. The case $D=N_T$ follows a similar idea; see the end of Section \ref{sec:4} for further discussions.

\section{Proof of the Main Result}\label{sec:4}
In this section, we provide a rigorous proof of Theorem \ref{theorem1}. The result is trivial when $L$ defined in \eqref{def:L} satisfies $L=N_I/2$, which results in a fully-connected RIS. In the following, we will focus on the case where $L=D=\sum_{k=1}^K N_k$.   The proof for $L=D=N_T$ follows a similar technique and  a brief discussion is given at the end of this section. In Section \ref{sec4a}, we formally state the more general result given in \eqref{generalresult} as Theorem \ref{theorem2}, which yields Theorem \ref{theorem1} as a direct corollary. The proof of Theorem \ref{theorem2} is then provided in Section \ref{sec4b}. 
\subsection{A More  General Result}\label{sec4a}
To begin, we introduce some notations and transformations.  Let $\bH_d=[\bH_{d,1}^T,\bH_{d,2}^T,\dots,\bH_{d,K}^T]^T\in\C^{\sum_{k=1}^K N_k\times N_T}$ and $\mathbf{H}_r=[\bH_{r,1}^T,\bH_{r,2}^T,\dots,\bH_{r,K}^T]^T\in\C^{\sum_{k=1}^K N_k\times N_I}$ denote the direct channel from the  transmitter to all users and the channel from the BD-RIS to all users, respectively. By introducing auxiliary variables $\{\bU_k\}_{k=1}^K$ as in \eqref{def:uk} and letting $\bU=[\bU_1,\bU_2,\dots,\bU_K]\in\C^{N_I\times \sum_{k=1}^KN_k}$, the sets of effective channels achieved by the fully-connected RIS and the BD-RIS architecture induced by graph $\mathcal{G}$, as defined in \eqref{Hfully} and \eqref{Hproposed}, respectively, can be expressed as follows:
\begin{equation}\label{Hfully2}
\begin{aligned}\mathcal{H}_{\text{fully}}=&\{\bH_d+\bU^{H}\bG\mid \bU\in\mathcal{U}_{\text{fully}}\},\end{aligned}
\end{equation}
where 
\begin{equation}\label{setU}
\begin{aligned}
\mathcal{U}_{\hspace{0.03cm}\text{fully}}\hspace{-0.05cm}=\hspace{-0.05cm}\{&\bU\in\C^{N_I\times \sum_{k=1}^KN_k}\mid\bU=(\bH_r\bthe)^{H}, 
~\bthe=\left(\mathbf{I}+\mathsf{i}Z_0\mathbf{B}\right)^{-1}\hspace{-0.05cm}\left(\mathbf{I}-\mathsf{i}Z_0\mathbf{B}\right),~ \bB=\bB^T \},
\end{aligned}
\end{equation}
and 
\begin{equation}\label{Hproposed2}
\begin{aligned}\mathcal{H}_{\mathcal{G}}=&\{\bH_d+\bU^{H}\bG\mid \bU\in\mathcal{U}_{\mathcal{G}}\},\end{aligned}
\end{equation}
where 
\begin{equation}\label{setU2}
\begin{aligned}
\mathcal{U}_{\mathcal{G}}=&\{\bU\in\C^{N_I\times \sum_{k=1}^K N_k}\mid\bU=(\bH_r\bthe)^{H}, 
\bthe\hspace{-0.05cm}=\hspace{-0.05cm}\left(\mathbf{I}\hspace{-0.05cm}+\hspace{-0.05cm}\mathsf{i}Z_0\mathbf{B}\right)^{-1}\hspace{-0.05cm}\left(\mathbf{I}\hspace{-0.05cm}-\hspace{-0.05cm}\mathsf{i}Z_0\mathbf{B}\right),~ \hspace{-0.05cm}\bB=\bB^T,~\bB\in\mathcal{B}_{\mathcal{G}}\}.
\end{aligned}
\end{equation}
In addition, given $\bU\in\C^{N_I\times \sum_{k=1}^K N_k}$, define  
\begin{equation}\label{Mu}
\bM_\bU\hspace{-0.05cm}=\hspace{-0.05cm}[\mathcal{R}\hspace{-0.05cm}\left(\mathsf{i}\hspace{-0.02cm}Z_0(\bH_r^H\hspace{-0.05cm}+\hspace{-0.05cm}\bU)\right)~\I\hspace{-0.05cm}\left(\mathsf{i}\hspace{-0.02cm}Z_0(\bH_r^H\hspace{-0.05cm}+\hspace{-0.05cm}\bU)\right)]\hspace{-0.05cm}\in\hspace{-0.05cm}\R^{N_I\times 2\hspace{-0.03cm}\sum_{k=1}^K\hspace{-0.08cm}N_k}\end{equation}
and let 
\begin{equation}\label{Nu}
\begin{aligned}
{\color{black}{\mathcal{N}}_\mathcal{U}=\left\{\bU\in\mathcal{U}_{\text{fully}}\mid \text{ there exists an } n\times n  \text{ singular submatrix of }  \bM_{\bU}, \text{where } n\leq2  \sum_{k=1}^K N_k\right\}.}
\end{aligned}
\end{equation}

We are now ready to explicitly establish the relationship in \eqref{generalresult}.

\begin{theorem}\label{theorem2}
Assume that $D=\sum_{k=1}^K N_k$. Let $\mathcal{G}$ be a graph whose adjacency matrix satisfies \eqref{adjacencymatrix} and \eqref{graph_condition} in Theorem \ref{theorem1}, then the sets $\mathcal{U}_{\text{\hspace{0.03cm}\normalfont fully}}$, $\mathcal{U}_{\mathcal{G}}$, and $\mathcal{N}_{\mathcal{U}}$ given in \eqref{setU}, \eqref{setU2}, and \eqref{Nu}, respectively, satisfy 
$$\mathcal{U}_{\hspace{0.03cm}\text{\normalfont fully}}=\mathcal{U}_{\mathcal{G}}\cup\mathcal{N}_{\mathcal{U}}.$$
This implies that the effective channel sets achieved by the fully-connected RIS and the BD-RIS architecture induced by $\mathcal{G}$,   given in \eqref{Hfully2} and \eqref{Hproposed2}, respectively, satisfy
$$\mathcal{H}_{\text{\normalfont fully}}=\mathcal{H}_{\mathcal{G}}\cup \mathcal{N},$$
where $\mathcal{N}=\left\{\bH_d+\bU^H\bG\mid\bU\in\mathcal{N}_{\mathcal{U}}\right\}$.  
\end{theorem}
{\color{black}The set $\mathcal{N}$ defined in Theorem \ref{theorem2} is a low-dimensional subset of $\mathcal{H}_{\text{fully}}$, or equivalently, $\mathcal{N}_\mathcal{U}$ is a low-dimensional subset of $\mathcal{U}_{\text{fully}}$.  To show this, let $N_\text{\text{sub}}$ be the total number of square submatrices of $\bM_\bU$,  and label these submatrices as $\bM_{\bU,1},\cdots,\bM_{\bU,N_{\text{sub}}}$. From the definition of $\mathcal{N}_{\mathcal{U}}$ in \eqref{Nu}, we have:
\begin{equation}\label{NUsubset}
{\mathcal{N}}_\mathcal{U}\subseteq\bigcup\limits_{i=1}^{N_{\text{sub}}}\{\bU\in{\mathcal{U}_{\text{fully}}}\mid \det\bM_{\bU,i}=0\}.
\end{equation}
Note that the set of all $n\times n$ singular matrices, i.e., $\{\mathbf{X}\in\R^{n\times n}\mid \det\mathbf{X}=0\},$
forms a low-dimensional subset of $\R^{n\times n}$. It follows that each set on the right-hand side of \eqref{NUsubset} is a low-dimensional subset of $\mathcal{U}_{\text{fully}}$, and so is their union. Hence, $\mathcal{N}_{\mathcal{U}}$ is a low-dimensional subset of $\mathcal{U}_{\text{fully}}$.

Theorem \ref{theorem1} is a direct corollary of Theorem \ref{theorem2} and the above claim; see below for the rigorous proof.}

\emph{Proof of Theorem \ref{theorem1}}: With the above notations, $P_{\text{fully}}$, i.e., the optimal value of \eqref{multiuser} achieved by fully-connected RIS, can be expressed as 
\begin{equation}\label{fully:problem}
\begin{aligned}
P_{\text{fully}}=\max_{\bU\in\mathcal{U}_{\hspace{0.03cm}\text{fully}},\,\bW\in\mathcal{W} }~&F(\bW, \bH_d+\bU^H\bG).
\end{aligned}
\end{equation}
Similarly, 
$$
\begin{aligned}
P_{\mathcal{G}}=\max_{\bU\in\mathcal{U}_{\mathcal{G}},\,\bW\in\mathcal{W} }~&F(\bW, \bH_d+\bU^H\bG).
\end{aligned}
$$
To guarantee that $P_{\text{fully}}=P_{\mathcal{G}}$, it suffices that there exists an optimal solution to \eqref{fully:problem}, denoted by $(\bU_{\text{\normalfont fully}}, \bW_{\text{\normalfont fully}})$, such that  $\bU_{\text{fully}}\in\mathcal{U}_{\mathcal{G}}.$  According to Theorem \ref{theorem2}, this is equivalent to $\bU_{\text{fully}}\notin\mathcal{N}_{\mathcal{U}}$.
Since $\mathcal{N}_{\mathcal{U}}$ is a low-dimensional subspace of $\mathcal{U}_{\text{fully}}$ and the problem parameter inputs, i.e., $\{\bG, \bH_{d},\bH_{r}\}$, are random in practice, $\bU_{\text{fully}}\notin\mathcal{N}_{\mathcal{U}}$ holds almost surely, which proves Theorem \ref{theorem1}. 


\subsection{Proof of Theorem \ref{theorem2}}\label{sec4b}
In this subsection, we give the proof of Theorem \ref{theorem2}.  Obviously, $\mathcal{U}_{\mathcal{G}}\cup\mathcal{N}_{\mathcal{U}}\subseteq\mathcal{U}_{\hspace{0.03cm}\text{\normalfont fully}}.$ We next prove $\mathcal{U}_{\hspace{0.03cm}\text{\normalfont fully}}\subseteq\mathcal{U}_{\mathcal{G}}\cup\mathcal{N}_{\mathcal{U}}.$  
For this purpose, we define an auxiliary set as follows:
$$\bar{\mathcal{U}}\hspace{-0.05cm}=\hspace{-0.05cm}\{\bU\hspace{-0.05cm}\in\hspace{-0.05cm}\C^{N_I\times \sum_{k=1}^K\hspace{-0.05cm}N_k}\hspace{-0.1cm}\mid\hspace{-0.05cm}\mathbf{U}^{H}\mathbf{U}\hspace{-0.05cm}=\hspace{-0.05cm}\mathbf{H}_r\mathbf{H}_r^H,\hspace{-0.05cm}~\mathbf{U}^{T}\mathbf{H}_r^H\hspace{-0.05cm}=\hspace{-0.05cm}(\mathbf{U}^{T}\mathbf{H}_r^H)^T\}.$$
In the following, we prove that 
$$\mathcal{U}_{\hspace{0.03cm}\text{fully}}\subseteq\bar{\mathcal{U}}\subseteq\mathcal{U}_{\mathcal{G}}\cup{\mathcal{N}}_\mathcal{U}.$$

\subsubsection{$\mathcal{U}_{\text{\hspace{0.03cm}\normalfont fully}}\subseteq\bar{\mathcal{U}}$}
Given     \(\mathbf{U}\in\mathcal{U}_{\hspace{0.03cm}\text{fully}}\), we next show that $\bU\in\bar{\mathcal{U}}$. Note that  $\bthe$ is unitary and symmetric due to the relationship between  $\bthe$ and $\bB$ in \eqref{theta_b} and the symmetry of $\bB$, i.e.,  $\bthe\bthe^{H}=\mathbf{I}$ and  $\bthe^T=\bthe$. Therefore, we have
$$\mathbf{U}^{H}\mathbf{U}=\mathbf{H}_r\bthe\bthe^{H}\mathbf{H}_r^{H}=\mathbf{H}_r\mathbf{H}_r^H$$
and

 $$
 \begin{aligned}\mathbf{U}^{T}\mathbf{H}_r^H&=\text{conj}(\bH_r)\text{conj}(\bthe)\mathbf{H}_r^H\\&=\text{conj}(\bH_r)\bthe^H\mathbf{H}_r^H=(\mathbf{U}^{T}\mathbf{H}_r^H)^T
 \end{aligned}$$
 where the second equality holds since $\text{conj}({\bthe})=\bthe^H$ as $\bthe=\bthe^T$. Hence, $\mathbf{U}\in\bar{\mathcal{U}}$.
\subsubsection{$\bar{\mathcal{U}}\subseteq\mathcal{U}_{\mathcal{G}}\cup{\mathcal{N}}_\mathcal{U}$}
It suffices to prove that for $\bU\in\bar{\mathcal{U}}\backslash\mathcal{N}_{\mathcal{U}}$, $\bU\in\mathcal{U}_{\mathcal{G}}$, i.e.,  there exists a solution $(\bthe,\bB)$ to the following system: 
\begin{equation*}
\left\{\begin{aligned}
&\bthe\bU=\bH_r^{H},\\
&\bthe=(\mathbf{I}+\mathsf{i}Z_0\mathbf{B})^{-1}(\mathbf{I}-\mathsf{i}Z_0\mathbf{B}),\\
&\bB=\bB^T,~\bB\in\mathcal{B}_{\mathcal{G}}.
\end{aligned}\right.
\end{equation*}
Since $\bB$ is real-valued,  the above equation can further be transformed into the real space as 
\begin{subequations}\label{constraint}
\begin{align}
&\bB\bM_{\bU}=\boldsymbol{\Gamma}_{\bU},\label{c1}\\
&\bB=\bB^T,~\bB\in\mathcal{B}_{\mathcal{G}}\label{c2},
\end{align}
\end{subequations}
where $\bM_{\bU}\in\R^{N_I\times 2\sum_{k=1}^KN_k}$ is given in \eqref{Mu} and 
\begin{equation}\label{gammaU}
\boldsymbol{\Gamma}_{\bU}=[\mathcal{R}(\bU-\bH_r^H)~~\I(\bU-\bH_r^H)]\in\R^{N_I\times 2\sum_{k=1}^KN_k}.
\end{equation} 
Note that for any permutation matrix $\bP\in\R^{N_I\times N_I}$,  $\bP^T\bB\bP$ is a solution to \eqref{c1} if and only if $\bB$ is a solution to \eqref{c1}, and thus we only need to consider graphs with adjacency matrices satisfying \eqref{graph_condition}. To simplify notation, we will focus on the case of band-connected RIS, where $\bB$ is given in \eqref{band} with $q=2L-1$.  The same proof technique applies to all other architectures. In the following proof, we  denote $\kappa=2L=2\sum_{k=1}^KN_k$ and $N=N_I$.

Similar to \cite{tree,wu}, we can reformulate \eqref{constraint} as an unconstrained linear system with respect to a new variable $\x$, which collects all the non-zero elements in the upper tridiagonal of $\bB$. In fact, $\x$ represents the actual degrees of freedom in $\bB$ constrained by \eqref{c2}.
Specifically,  define 
$$\begin{aligned}
\x:=[\bB(1,1:\kappa), ~&\bB(2, 2: \kappa+1),~ \dots, ~\bB(N-1, N-1:N),~ \bB(N,N)]^T,
\end{aligned}$$ and let $\{\ba_n^T\}_{n=1}^N$ and $\{\bb_n^T\}_{n=1}^N$ be the rows of $\bM_{\bU}$ and $\boldsymbol{\Gamma}_{\bU}$, respectively, i.e.,  
\begin{equation}\label{def:ab}
\bM_{\bU}^T=[\mathbf{a}_1,\ba_2,\dots,\ba_N],~~\boldsymbol{\Gamma}_{\bU}^T=[\bb_1,\bb_2,\dots,\bb_N].
\end{equation} Then \eqref{constraint} can be equivalently expressed  as the following linear equation \cite{wu}:
\begin{equation}\label{Ax-b}
\bA\x=\bb,
\end{equation} where 
 \begin{equation}\label{splitA}\bA=\left[\begin{matrix}\bA_{1,1}&\bA_{1,2}&\cdots&\bA_{1,N}\\\bA_{2,1}&\bA_{2,2}&\cdots&\bA_{2,N}\\\vdots&\vdots&\ddots&\vdots\\\bA_{N,1}&\bA_{N,2}&\cdots&\bA_{N,N}
 \end{matrix}\right]\in\R^{N\kappa\times \left(N\kappa-\frac{\kappa(\kappa-1)}{2}\right)}
 \end{equation} 
 with $\bA_{m,n}\in\R^{\kappa\times \min\{\kappa, N-n+1\}}$ given by
\begin{equation}\label{def:Ai}
\bA_{m,n}(:,q)=\left\{
\begin{aligned} 
\ba_{q+n-1},~~\,&\text{if }m=n;\\
\ba_n,~~~~~~~ &\text{if }m=q+n-1;\\
\mathbf{0},~~~~~~~~&\text{otherwise},
\end{aligned}\right.
\end{equation}
and 
\begin{equation}\label{def:b}
\bb=\text{vec}\left(\boldsymbol{\Gamma}_{\bU}^T\right)=[\bb_1^T,\bb_2^T,\dots,\bb_N^T]^T.
\end{equation}
For clarity, we give the structure of $\bA$ in  Fig. \ref{A:structure}.
\begin{figure*}
\includegraphics[scale=0.3]{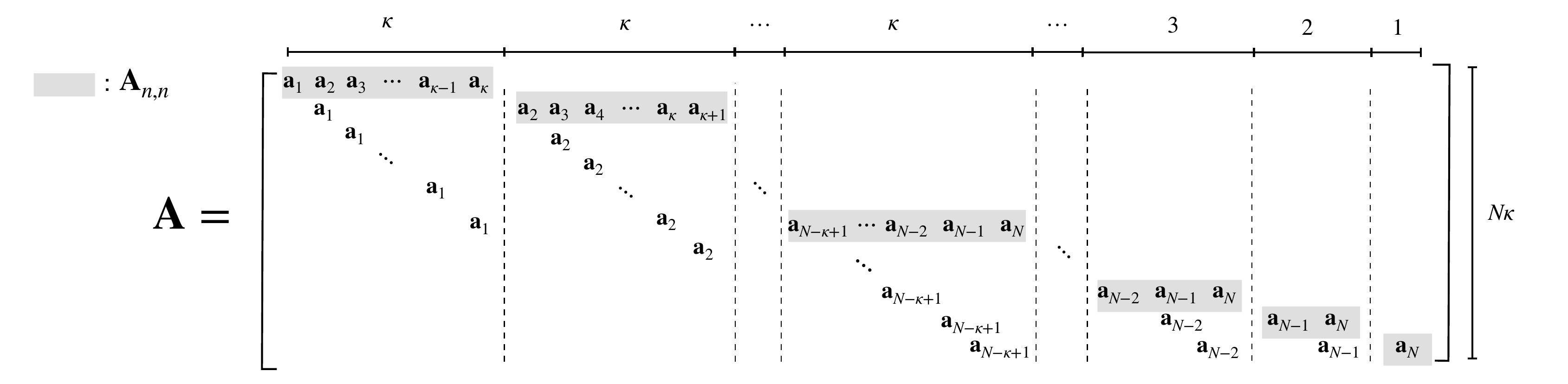}
\centering
\caption{The structure of $\bA$ in \eqref{Ax-b}, where the shaded blocks represent $\bA_{n,n}$, $n=1,2,\dots, N$.}
\label{A:structure}
\centering
\end{figure*}
The remaining task is to show that there exists a solution to the linear system $\bA\x=\bb$. According to the Rouche-Capelli theorem \cite{strang2022introduction}, it suffices to prove that $\text{\normalfont{rank}}(\bA)=\text{\normalfont{rank}}([\bA,\bb])$. We relegate the remaining proof to  Lemma \ref{lemmaA} in Appendix \ref{app:lemmaA} and Lemma \ref{lemmaAb} in Appendix \ref{app:lemmaAb}, where we prove that $$\text{\normalfont{rank}}(\bA)=\text{\normalfont{rank}}([\bA,\bb])=N\kappa-\frac{\kappa(\kappa-1)}{2}.$$ 

Finally, we note that when $D=N_T$, we can introduce an auxiliary variable $\bV=\bthe\bG\in\C^{N_I\times N_T}$ to express the effective channel sets as 
\begin{equation*}
\begin{aligned}\mathcal{H}_{\text{fully}}=&\{\bH_d+\bH_r\bV\mid \bV\in\mathcal{V}_{\text{fully}}\}\end{aligned}
\end{equation*}
and 
\begin{equation*}
\begin{aligned}\mathcal{H}_{\mathcal{G}}=&\{\bH_d+\bH_r\bV\mid \bV\in\mathcal{V}_{\mathcal{G}}\},
\end{aligned}
\end{equation*}
where 
\begin{equation*}
\begin{aligned}
\mathcal{V}_{\hspace{0.03cm}\text{fully}}\hspace{-0.05cm}=\hspace{-0.05cm}\{&\bV\in\C^{N_I\times N_T}\mid\bV=\bthe\bG, 
~\bthe=\left(\mathbf{I}+jZ_0\mathbf{B}\right)^{-1}\hspace{-0.05cm}\left(\mathbf{I}-jZ_0\mathbf{B}\right),~ \bB=\bB^T \}
\end{aligned}
\end{equation*}
and 
\begin{equation*}
\begin{aligned}
\mathcal{V}_{\mathcal{G}}=&\{\bV\in\C^{N_I\times N_T}\mid\bV=\bthe\bG, ~
\bthe\hspace{-0.05cm}=\hspace{-0.05cm}\left(\mathbf{I}\hspace{-0.05cm}+\hspace{-0.05cm}jZ_0\mathbf{B}\right)^{-1}\hspace{-0.05cm}\left(\mathbf{I}\hspace{-0.05cm}-\hspace{-0.05cm}jZ_0\mathbf{B}\right),~ \hspace{-0.05cm}\bB=\bB^T, \bB\in\mathcal{B}_{\mathcal{G}}\}.
\end{aligned}
\end{equation*}
Following the same ideas as previous discussions, we can establish a relationship between  $\mathcal{V}_{\text{fully}}$ and $\mathcal{V}_{\mathcal{G}}$ similar to that between $\mathcal{U}_{\text{fully}}$ and $\mathcal{U}_{\mathcal{G}}$  in Theorem \ref{theorem2}, which immediately yields  Theorem \ref{theorem1}. 

\section{Simulation Results}\label{sec:5}
In this section, we present simulation results to validate our theoretical results and demonstrate the effectiveness of the  proposed architectures. 

We consider a BD-RIS-aided multiuser MIMO system, where the distances between the transmitter and the BD-RIS, and between the BD-RIS and the users, are set to $d_{IT}=50$\,m and $d_{RI}=2.5$\,m, respectively.  For simplicity, we assume that the direct channel between the transmitter and the users is blocked  and each user is equipped with a single antenna, i.e., $N_k=1,~ k=1,2,\dots, K$, as in \cite{wu}. In addition, we set $K\leq N_T\leq N_I/2$, and hence $L=D=K$.
For the channels from  the transmitter to the BD-RIS and from the BD-RIS to the receivers, we  employ a distance-dependent pathloss model $L(d)=L_0 d^{-\alpha}$ to account for large-scale fading, where $L_0$ is the reference path loss at distance $d_0=1$\,m,  $d$ is the distance, and $\alpha$ is the path loss exponent. In our simulation, we set $L_0=-30$\,dB and $\alpha=2.2$. Regarding the small-scale fading, {\color{black} unless otherwise stated, }we adopt the Rician fading model, with the Rician factor set to $\kappa=2$\,dB. The noise power at each user is $\sigma^2=-80$\,dBm. 
To evaluate the performance of our proposed architectures, we consider two specific metrics: sum-rate maximization and transmit power minimization under QoS constraints. The power budget for sum-rate maximization is $P=10$\,dBm, and the QoS constraint in  transmit power minimization is $\text{SINR}_k\geq\gamma_k=\gamma=15$\, dB. 
The corresponding optimization problems are solved using the algorithms proposed in our recent work \cite{wu}, {\color{black} which are applicable to arbitrary BD-RIS architectures.}

In Fig. \ref{performance}, we validate our theoretical results in Theorem \ref{theorem1}. Specifically, we compare the performance of fully-connected RIS with that of the two optimal architectures introduced in 
Section \ref{sec:32}, i.e.,  band-connected RIS with band width $q=2L-1$  and stem-connected RIS with stem width $q=2L-1$. For reference,  the single-connected RIS and tree-connected RIS (tridiagonal RIS and arrowhead RIS) are also included. As shown in the figure,  the band/stem-connected RIS with band/stem width $q=2L-1$ exhibits exactly the same performance as fully-connected RIS,  which verifies Theorem \ref{theorem1}. In contrast, the tree- and single-connected RIS suffer from a severe performance degradation compared to fully-connected RIS.  It is interesting to note that, while both the tridiagonal RIS and arrowhead RIS are optimal in single-user MISO systems, the tridiagonal RIS demonstrates slightly better performance in the multiuser scenario.
\begin{figure}[t]
\subfigure[Sum-rate maximization.]{
\includegraphics[width=0.45\columnwidth]{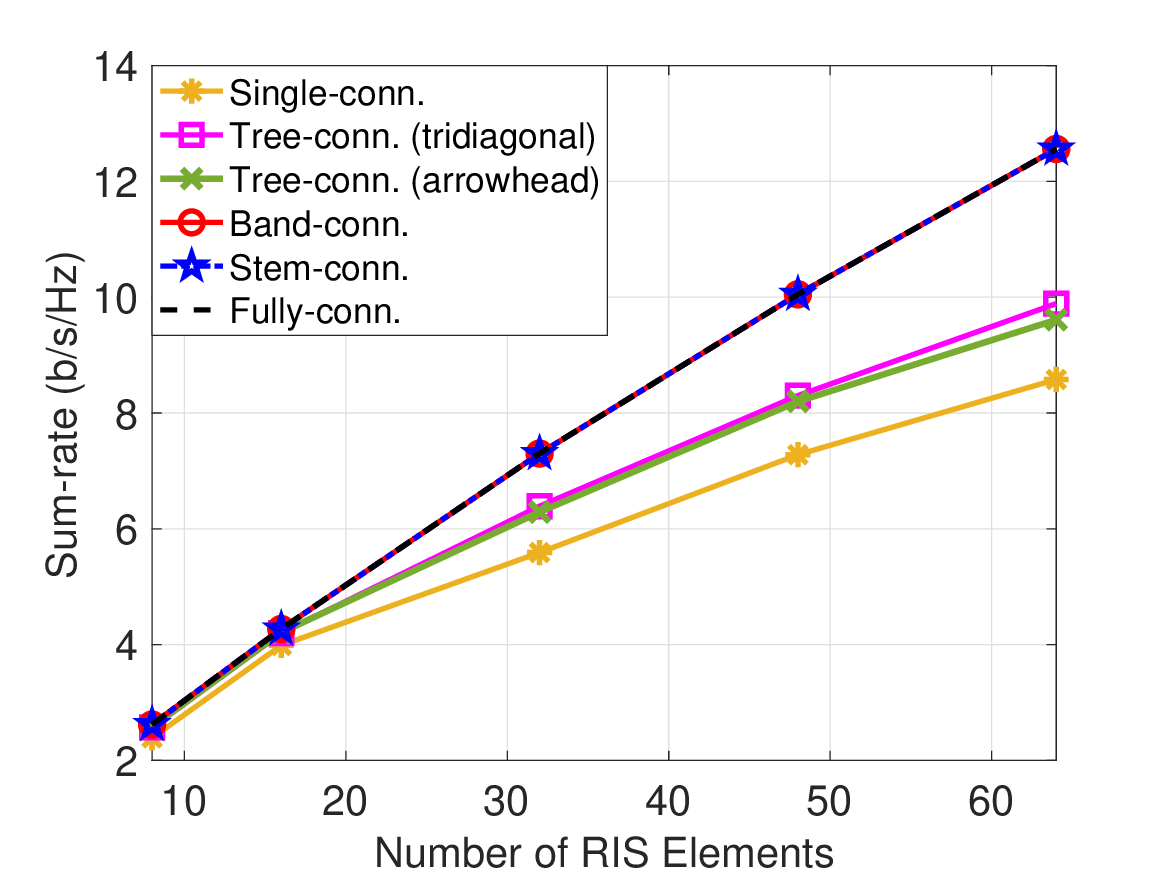}}
\subfigure[Transmit power minimization.]{
\includegraphics[width=0.45\columnwidth]{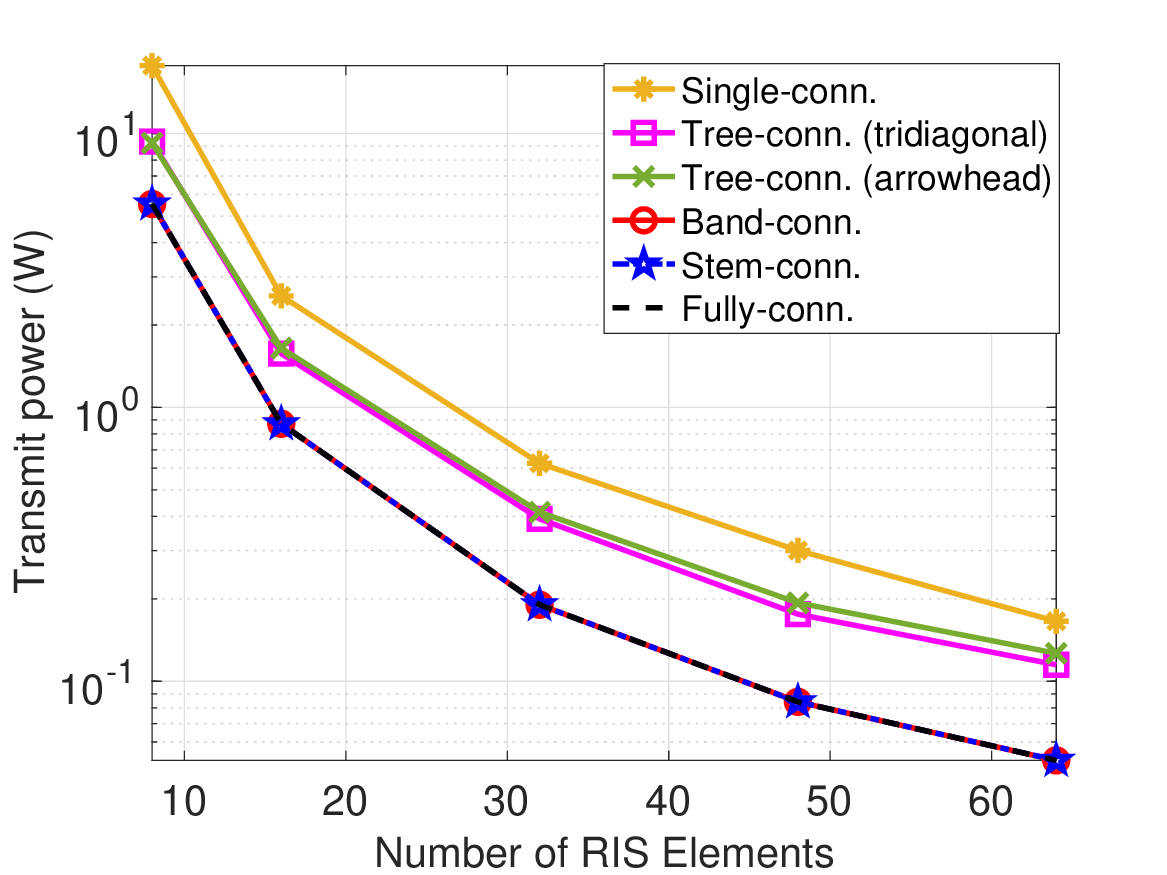}}
\centering
\caption{Performance of different BD-RIS architectures versus number of RIS elements ($N_T=4, K=4$).}
\label{performance}
\end{figure}

In Fig. \ref{complexity}, we further illustrate the circuit complexity of the aforementioned architectures in terms of the required number of admittances.  Both the band-connected and stem-connected RIS with band and stem width $q=2L-1$ belong to the architecture class described in Theorem \ref{theorem1} and therefore exhibit the same circuit complexity.  These architectures, along with all others specified by Theorem \ref{theorem1}, achieve a significant reduction in circuit complexity compared to fully-connected RIS, especially when the number of RIS elements is large.
 For example, the number of admittances is reduced by more than fourfold when $N_I=64$. On the other hand, compared to single- and tree-connected RIS, the proposed architectures exhibit a slightly higher circuit complexity but achieve significantly improved performance. 
 
\begin{figure}[t]
\includegraphics[width=0.45\columnwidth]{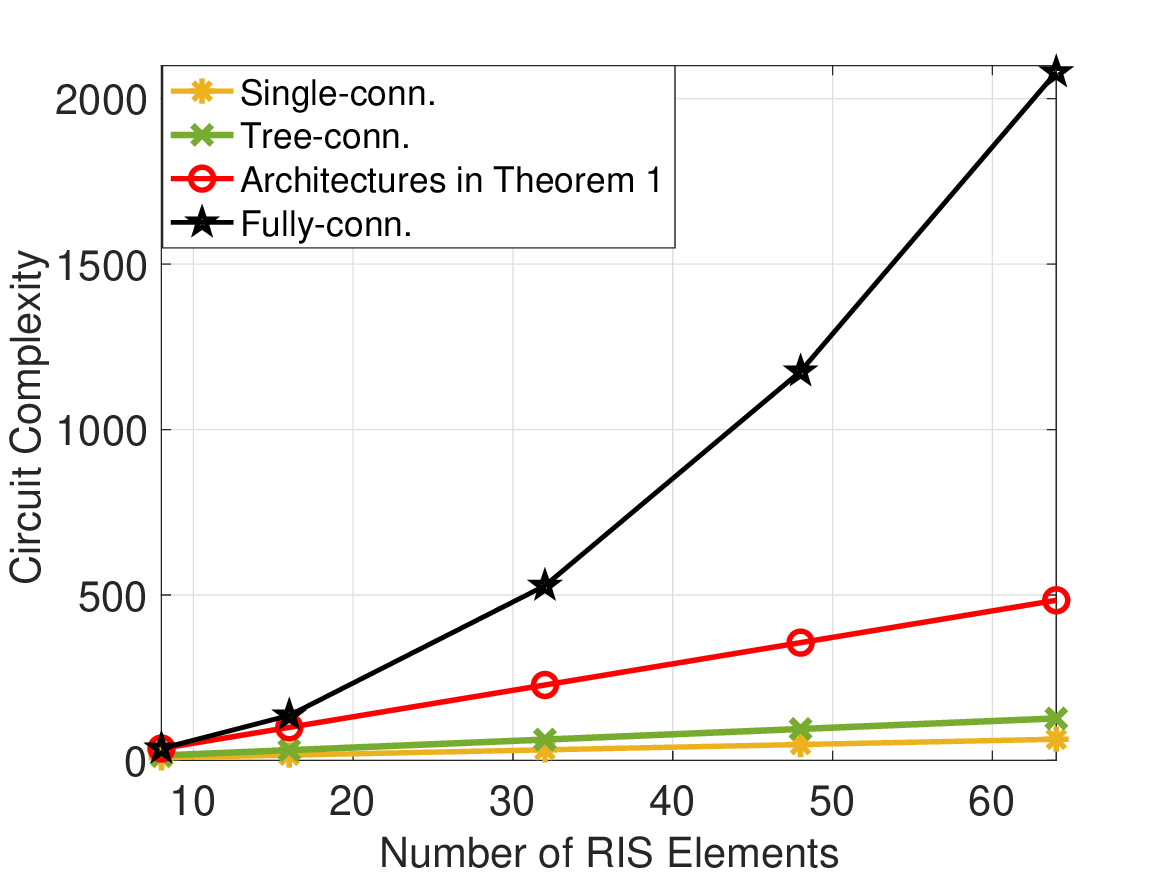}
\centering
\caption{Circuit complexity (i.e., the number of admittances) for different BD-RIS architectures versus number of RIS elements ($N_T=4, K=4$).}
\label{complexity}
\end{figure}
\begin{figure}[t]
\subfigure[Sum-rate maximization.]{
\includegraphics[width=0.45\columnwidth]{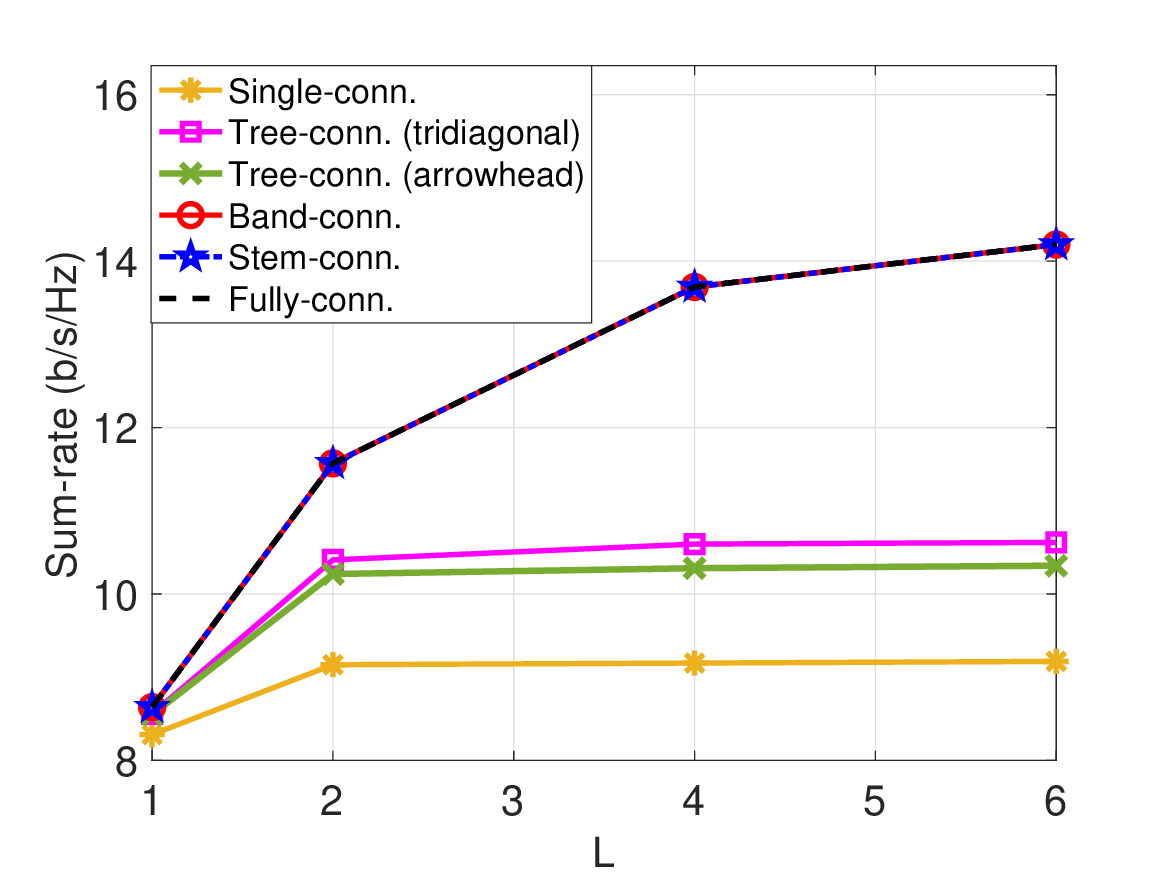}}
\subfigure[Transmit power minimization.]{
\includegraphics[width=0.45\columnwidth]{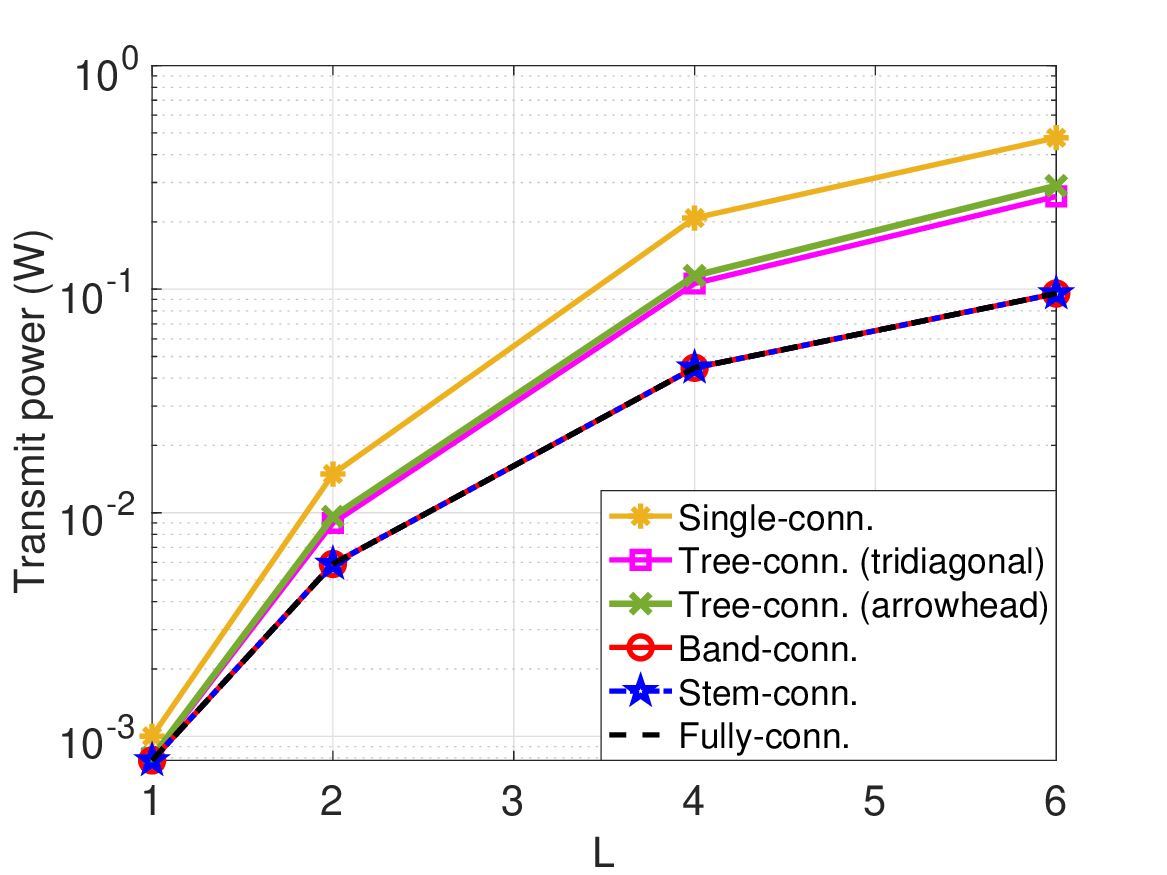}}
\centering
\caption{Performance of different BD-RIS architectures versus $L=K$ ($N_I=64, N_T=6$).}
\label{performance2}
\end{figure}
\begin{figure}[t]
\includegraphics[width=0.5\columnwidth]{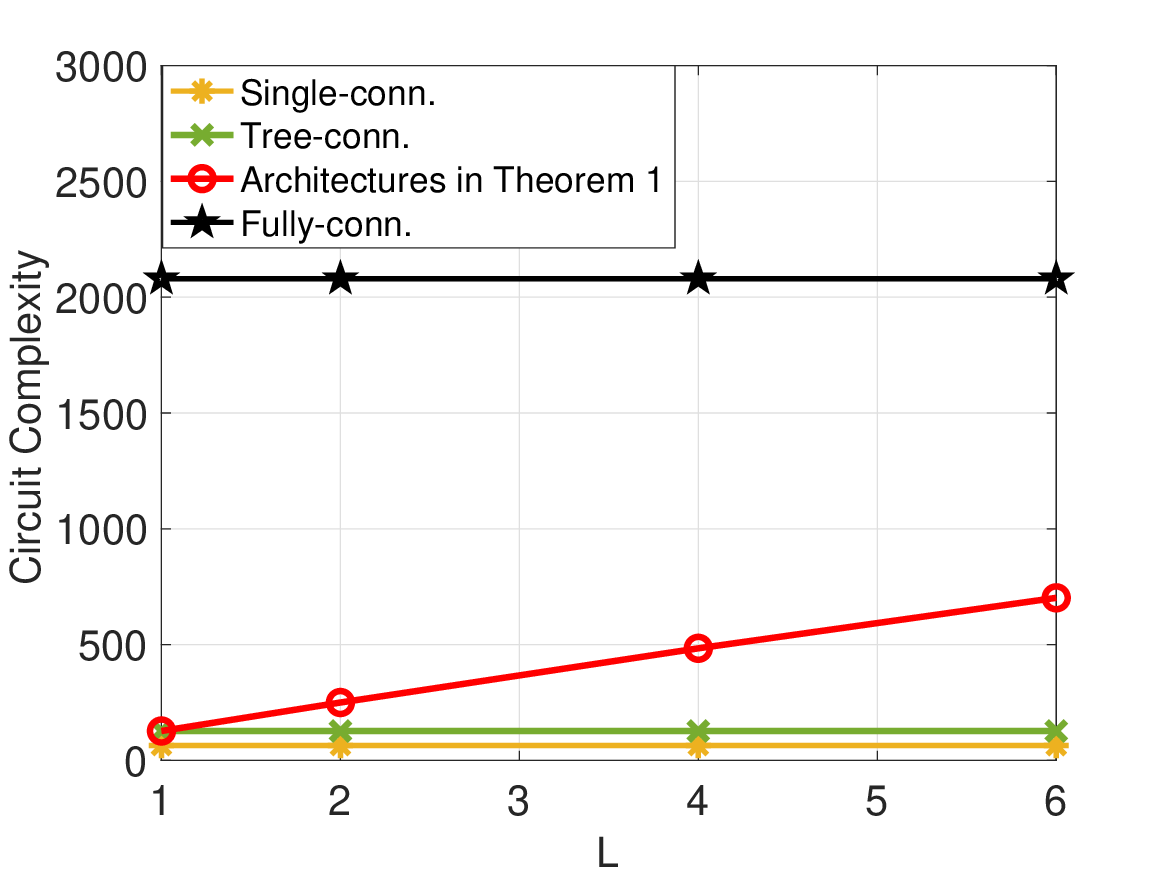}
\centering
\caption{Circuit complexity (i.e., the number of admittances) for different BD-RIS architectures versus $L=K$ ($N_I=64, N_T=6$).}
\label{complexity2}
\end{figure}

In Figs. \ref{performance2} and \ref{complexity2},  we examine how the value of $L$ affects the performance and circuit complexity of different BD-RIS architectures. We fix the numbers of RIS elements and transmit antennas as $N_I=64$ and $N_T=6$, respectively,  and increase the number of users from $K=1$ to $K=6$, such that $L=K$.  
As can be observed, the performance superiority of band-, stem-, and fully-connected RIS over tree- and single-connected RIS is more prominent as $L$ increases. To achieve full flexibility of BD-RIS, however,  the circuit complexity of band- and stem-connected RIS (as well as the architecture class given in Theorem \ref{theorem1}) also grows with $L$, whereas the circuit complexity for all the other architectures is fixed when the number of RIS elements $N_I$ is fixed. Nevertheless, the circuit complexity of our proposed architectures remains significantly lower than that of fully-connected RIS.

{\color{black}Another interesting observation from Fig. \ref{performance2} (a)  is that the sum-rate performance of single-connected RIS quickly saturates as $L$ (i.e., $K$) increases, indicating that it is unable to exploit the presence of multiple users and boost the sum-rate performance. To investigate the underlying reason, we replace the Rician fading channel model in Fig. \ref{performance2} (a) with a Rayleigh fading model  in Fig. \ref{sumrate_K_rayleigh}. Comparing Fig. \ref{performance2} (a) and Fig. \ref{sumrate_K_rayleigh},  we observe that under Rayleigh fading,  the performance of all architectures saturates more slowly than that under Rician fading.  
This is because, compared to Rician fading, Rayleigh fading exhibits richer scattering and greater channel variability among the users, which reduces multiuser interference and enhances multiuser diversity. 
In both fading models, the sum-rate of single-connected RIS grows more slowly with $K$ than in BD-RIS due to the lack of interconnections. This observation demonstrates the crucial role of BD-RIS in efficiently serving an increasing number of users.}

\begin{figure}[t]
\includegraphics[width=0.45\columnwidth]{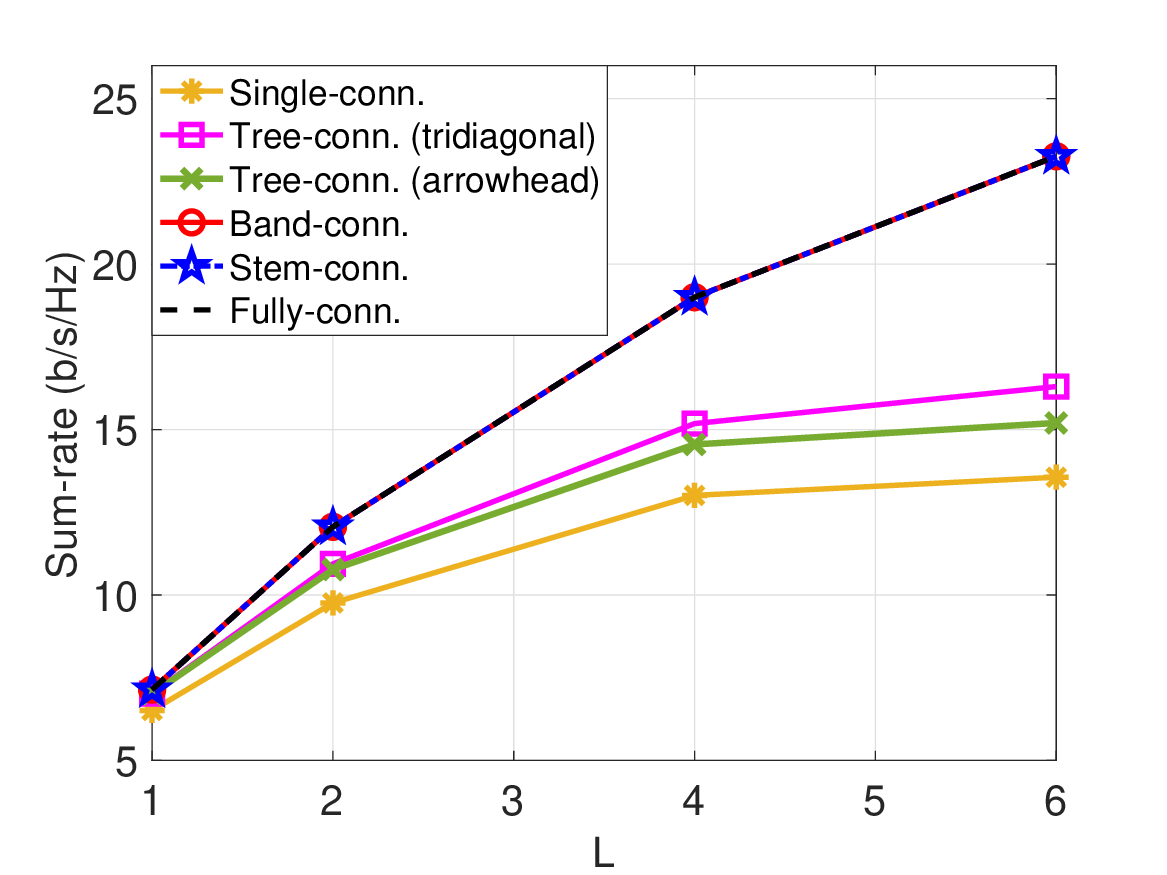}
\centering
\caption{Sum-rate of different BD-RIS architectures versus $L=K$ for Rayleigh fading channel ($N_I=64, N_T=6$).}
\label{sumrate_K_rayleigh}
\end{figure}

In Fig. \ref{pareto}, we compare the performance of the band- and stem-connected architectures proposed in Section \ref{sec:32} with that of the classical group-connected architecture introduced in Section \ref{existing_arch} by plotting their performance-complexity Pareto frontiers, i.e., the achievable sum-rate or transmit power versus circuit complexity. For the group-connected architecture, the number of RIS elements in each group is identical, denoted as $G_s:=N_I/G$. The curves are obtained by gradually increasing the corresponding group size $G_s$ and band/stem width $q$.
From Fig. \ref{pareto}, we can make the following observations.  First, as illustrated earlier, the band- and stem-connected RIS with band and stem width $q=2L-1$, which satisfy  the conditions in Theorem \ref{theorem1}, are able to achieve the same performance as fully-connected RIS with significantly lower circuit complexity. Second, the proposed band- and stem-connected architectures consistently outperform the classical group-connected architecture, offering higher sum rate and lower transmit power with the same circuit complexity.

 Third, among the two proposed architectures, band-connected RIS achieves slightly better performance than stem-connected RIS at low circuit complexity.

\begin{figure}[t]
\subfigure[Sum-rate maximization.]{
\includegraphics[width=0.45\columnwidth]{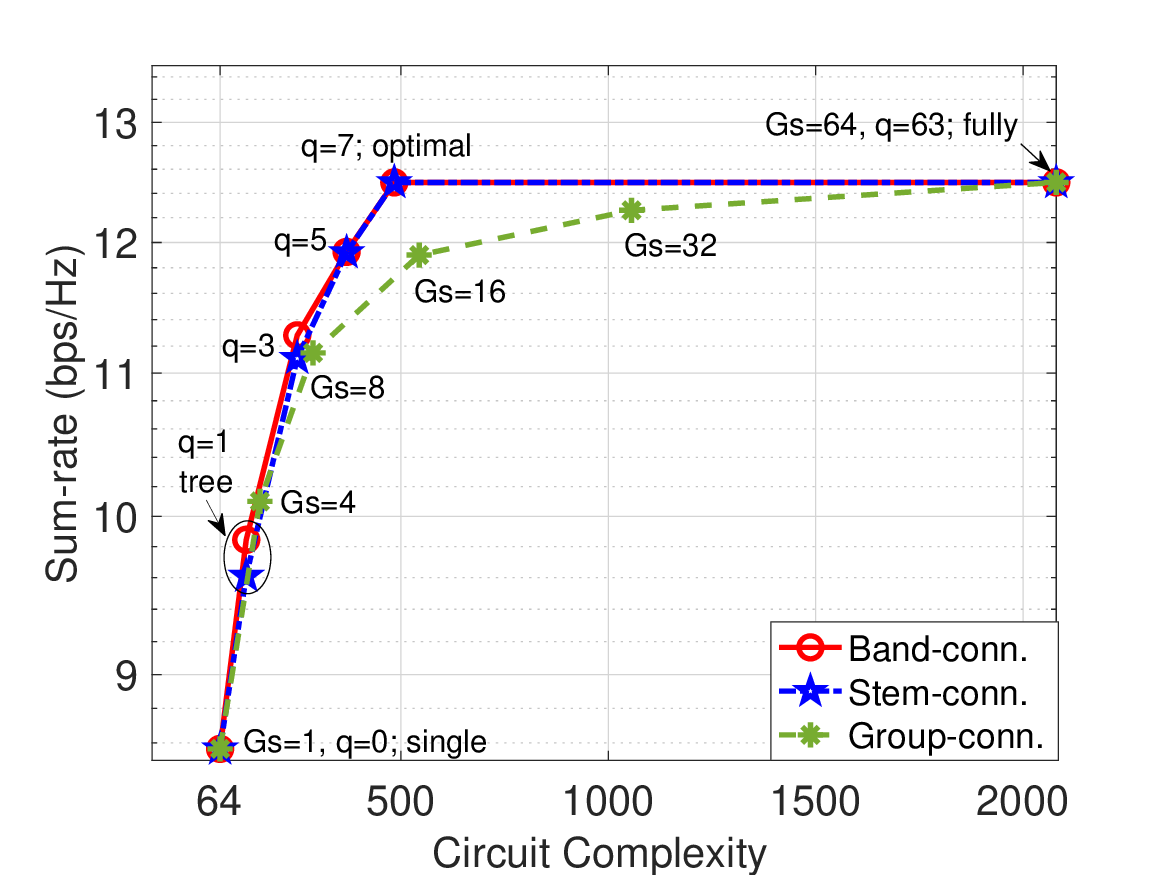}}
\subfigure[Transmit power minimization.]{
\includegraphics[width=0.45\columnwidth]{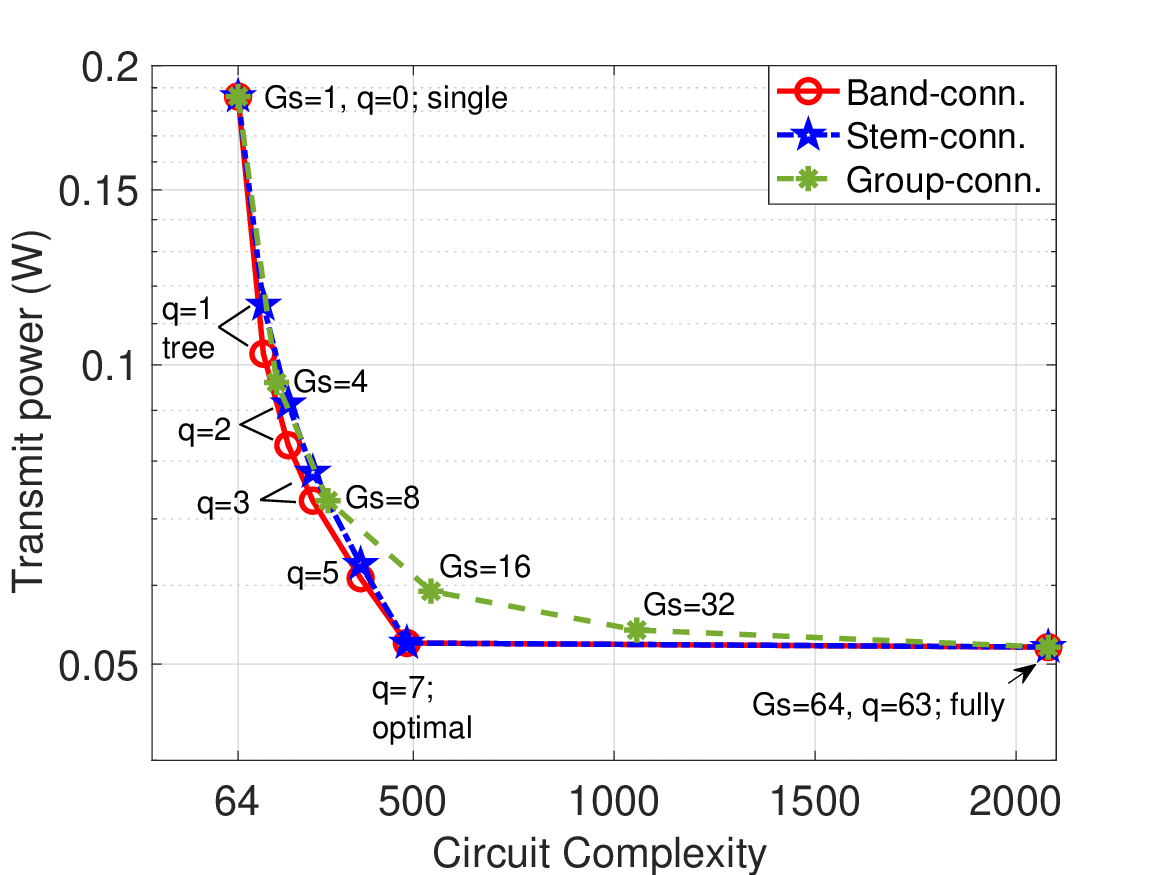}}
\centering
\caption{Pareto frontier between performance and circuit complexity achieved by group-, stem-, and band-connected RIS, with the corresponding group size $G_s$ and stem/band width $q$ labeled on the right and left sides of the curves, respectively ($N_I=64, N_T=4, K=4$).}
\label{pareto}
\end{figure}

To conclude, the simulation results demonstrate that our proposed BD-RIS architectures achieve a better trade-off between performance and circuit complexity compared to existing  architectures. Specifically, in terms of the optimal performance, the architectures specified by Theorem \ref{theorem1}, and particularly the band- and stem-connected RIS with band and stem width $q=2L-1$, match the performance of fully-connected RIS while significantly simplifying circuit complexity. Furthermore,  the band- and stem-connected RIS achieve better performance-complexity Pareto frontiers than the classical group-connected architecture, offering superior performance in general cases.


\section{Conclusion}\label{sec:6}
In this paper, we investigate the architectural design of BD-RIS. We model BD-RIS architectures using graph theory and characterize their topological connectivity through the adjacency matrix of the graph. As a key theoretical contribution (i.e., Theorem \ref{theorem1}), we propose new BD-RIS architectures that match the performance of fully-connected RIS with  significantly reduced circuit complexity. Our result is general in the following sense:  First, it provides a general condition on the adjacency matrix that encompasses a class of BD-RIS architectures. Second, it  applies to general multiuser MIMO systems, which includes the result in \cite{tree} for single-user MISO systems as a special case. Third, the optimality of the proposed architectures holds regardless of the specific performance metric. We further introduce two novel BD-RIS architectures, namely band-connected RIS and stem-connected RIS,  and show that they belong to the optimal architecture class under certain conditions.  Numerical results validate our theoretical result and demonstrate the superiority of the band- and stem-connected architectures over the classical group-connected architecture as well as their significant performance gains over conventional RIS.

One limitation of our current result is that it only provides a sufficient condition for achieving the optimal performance. An important  future work is to identify the architecture that achieves the optimal performance with the least circuit complexity in multiuser MIMO systems. 

\appendices 
\section{Proof of Remark \ref{remark:tree}}\label{app:tree}
In this appendix, we give the proof of Remark \ref{remark:tree}. The following lemma presents two properties of tree graphs that will be used in our proof.
 \begin{lemma}[{\cite{bondy2008graph}}]
The following results hold:
\begin{itemize}\label{lemma:tree}
\item[(a)] A graph $\mathcal{G}$ is a tree if and only if $\mathcal{G}$ is connected and has $N-1$ vertices, where $N$ is the number of vertices in $\mathcal{G}$.
\item[(b)] Any tree with $N>1$ vertices has at least two vertices of degree one, which are referred to as the leaves of the tree.  \end{itemize}
\end{lemma}
 Next we prove Remark \ref{remark:tree}. Note that when $K=N_k=1$, \eqref{graph_condition} reduces to 
\begin{equation}\label{graph_condition2}
\sum_{m=n+1}^{N_I}\bar{A}_{\mathcal{G}}(n,m)=1,~~1\leq n\leq N_I-1.
\end{equation}
The goal is to show that $\mathcal{G}$ is a tree if and only if its adjacency matrix satisfies \eqref{adjacencymatrix} and \eqref{graph_condition2}. This can be proved by induction. Obviously, when $N_I=2$, $\mathcal{G}$ is a tree if and only if its adjacency matrix has the following form $$\bA_{\mathcal{G}}=\left[\begin{matrix} 0&1\\1&0\end{matrix}\right],$$
i.e., \eqref{adjacencymatrix} and \eqref{graph_condition2} are satisfied with $\mathbf{P}=\mathbf{I}_2$. Now suppose that the assertion holds for $N_I=2,3,\dots, i-1$.  We next prove the assertion for $N_I=i$. 

Given  a tree $\mathcal{G}$ with $i>2$ vertices, we next show that its adjacency matrix satisfies \eqref{adjacencymatrix} and \eqref{graph_condition2}. Let $V_n$ be a vertex with degree one, and denote the corresponding edge as $e_{n,m}=(V_n,V_m)$. Such  a vertex $V_n$ must exist according to Lemma \ref{lemma:tree} (b). 
  Let $\mathcal{G}_{\backslash n} $ be the subgraph of $\mathcal{G}$ obtained by removing $V_n$. Then, $\mathcal{G}_{\backslash n}$ is a connected graph with $i-1$ vertices and $i-2$ edges, which is a tree according to Lemma \ref{lemma:tree} (a).  By the inductive hypothesis, there exists a permutation matrix $\mathbf{P}_{\backslash n}\in\R^{(i-1)\times (i-1)}$ such that the adjacency matrix of $\mathcal{G}_{\backslash n}$, denoted $\bA_{\mathcal{G}_{\backslash n}}$,  can be expressed as 
$$\bA_{\mathcal{G}_{\backslash n}}=\mathbf{P}_{\backslash n}\bar{\bA}_{\mathcal{G}_{\backslash n}}\mathbf{P}_{\backslash n}^T,$$
where $\bar{\bA}_{\mathcal{G}_{\backslash n}}\in\R^{(i-1)\times (i-1)}$ satisfies \eqref{graph_condition2} with $N_I=i-1$. Let $\bP_{1,n}\in\R^{i\times i}$ be the permutation matrix obtained by swapping the $1$-st and $n$-th columns of $\mathbf{I}_i$. The adjacency matrix of $\mathcal{G}$ can then be expressed as
\begin{equation}\label{A_G}
\begin{aligned}
\mathbf{A}_{\mathcal{G}}&=\mathbf{P}_{1,n}\left[\begin{matrix}0&\mathbf{e}_{l}^T\\\mathbf{e}_{l}&\bA_{\mathcal{G}_{\backslash n}}
\end{matrix}\right]\mathbf{P}_{1,n}^T\\
&=\underbrace{\mathbf{P}_{1,n}\left[\begin{matrix}1&0\\0&\mathbf{P}_{\backslash n}\end{matrix}\right]}_{\bP}\underbrace{\left[\begin{matrix}0&(\mathbf{P}_{\backslash n}\mathbf{e}_{l})^T\\\mathbf{P}_{\backslash n}\mathbf{e}_{l}&\bar{\bA}_{\mathcal{G}_{\backslash n}}
\end{matrix}\right]}_{\bar{\bA}_{\mathcal{G}}}\underbrace{\left[\begin{matrix}1&0\\0&\mathbf{P}_{\backslash n}\end{matrix}\right]^T\mathbf{P}_{1,n}^T}_{\bP^T},
\end{aligned}
\end{equation}
where $\mathbf{e}_l\in\R^{i-1}$, with $l=m-1$ if $m>n$, and $l=m$ if $m<n$. 
Clearly, $\bar{\bA}_{\mathcal{G}}$ satisfies \eqref{adjacencymatrix} and  \eqref{graph_condition2} with $N=i$. 

On the other hand, let $\mathcal{G}$ be a graph whose adjacency matrix satisfies \eqref{adjacencymatrix} and  \eqref{graph_condition2} with $N=i$. We next show that $\mathcal{G}$ is a tree.  It suffices to prove this assertion for $\bP=\mathbf{I}_i$, as the condition in \eqref{adjacencymatrix} defines a class of isomorphic graphs. 
   To show this, let ${\mathcal{G}}_{\backslash 1}$ be the subgraph of ${\mathcal{G}}$ obtained by deleting its first vertex ${V}_1$. Obviously, the adjacency matrix of ${\mathcal{G}}_{\backslash 1}$ satisfies \eqref{adjacencymatrix} and  \eqref{graph_condition2} with $\bP=\mathbf{I}_{i-1}$ and $N=i-1$, and hence ${\mathcal{G}}_{\backslash 1}$ is a tree  by the inductive hypothesis.  Note that $V_1$ in ${\mathcal{G}}$ has degree one since 
$$\sum_{m=2}^iA_{\mathcal{G}}(1,m)=1. $$ This, together with the fact that ${\mathcal{G}}_{\backslash 1}$ is a tree, implies that ${\mathcal{G}}$ is connected and has $i-1$ edges. Therefore, ${\mathcal{G}}$ is a tree by Lemma \ref{lemma:tree} (a), which completes the proof.
\section{Rank of matrix $\bA$ given in \eqref{splitA}.}\label{app:lemmaA}
In this appendix, we investigate the rank of $\bA$ given in \eqref{splitA}. 
\begin{lemma}\label{lemmaA}
$\text{\normalfont{rank}}(\bA)=N\kappa-\frac{\kappa(\kappa-1)}{2}.$
\end{lemma}
\begin{proof}
According to \eqref{def:Ai}, it is easy to see that $\bA_{n,m}=\mathbf{0}$ if $n<m$ (see Fig. \ref{A:structure}). Hence,  $\bA$ can be expressed as a block lower triangular matrix as
$$\bA=\left[\begin{array}{cccc}
\bA_{1,1}&&&\\
\bA_{2,1}&\bA_{2,2}&&\\
\vdots&\vdots&\ddots&\\
\bA_{N,1}&\bA_{N,2}&\cdots&\bA_{N,N}
\end{array}\right].$$
Since $\bU\notin\mathcal{N}_{\bU}$,  any $n$ columns of $\{\ba_1,\ba_2,\dots,\ba_N\}$ with $n\leq \kappa$ are linear independent. Note that
$$\bA_{n,n}=
\left\{
\begin{aligned}[\ba_n,\ba_{n+1},\dots, \ba_{n+\kappa-1}],~~&\text{if}~ 1\leq n\leq N-\kappa+1;\\
[\ba_n,\ba_{n+1},\dots,\ba_{N}],~~~~~~~&\text{if}~N-\kappa+1<n\leq N,
\end{aligned}\right.$$ which is of full column rank due to the fact that the number of columns in $\bA_{n,n}$ is no more than $\kappa$.  Therefore, we have
$$\text{rank}(\bA_{n,n})=\left\{
\begin{aligned}
\kappa,~~~~\,~~~&\text{if}~ 1\leq n\leq N-\kappa+1;\\
N-n+1,~~&\text{if}~N-\kappa+1<n\leq N.
\end{aligned}\right.$$
This, together with the lower triangular structure of $\bA$, gives 
 \begin{equation}\label{rankAgeq}
 \hspace{-0.1cm}\begin{aligned}    
\text{rank}(\bA)&\geq\sum_{n=1}^N \text{rank}(\bA_{n,n})\\
&=(N-\kappa+1)\kappa+\kappa-1+\kappa-2+\dots+1\\
&= N\kappa-\frac{\kappa(\kappa-1)}{2}.
\end{aligned}
\end{equation}
On the other hand, since $\bA\in\R^{N\kappa\times \left(N\kappa-\frac{\kappa(\kappa-1)}{2}\right)}$, we immediately have
\begin{equation}\label{rankAleq}\text{rank}(\bA)\leq N\kappa-\frac{\kappa(\kappa-1)}{2}.\end{equation}
Combining \eqref{rankAgeq} and \eqref{rankAleq} gives the desired result. 
\end{proof}
\section{Rank of matrix $\left[\bA~\bb\right]$.}\label{app:lemmaAb}
In this appendix, we characterize the rank of matrix $\left[\bA~\bb\right]$, where $\bA$ and $\bb$ are given in \eqref{splitA} and \eqref{def:b}, respectively. Specifically, we prove the following result. 
\begin{lemma}\label{lemmaAb}
$\text{\normalfont{rank}}([\bA~\bb])=N\kappa-\frac{\kappa(\kappa-1)}{2}$.
\end{lemma}
Before going into detailed proof of Lemma \ref{lemmaAb}, we first give two auxiliary results in Appendix \ref{auxiliary}, which are important to the proof.  The complete proof of Lemma \ref{lemmaAb} is provided in Appendix \ref{proof:lemmaAb}.
\subsection{Auxiliary results.}\label{auxiliary}
   \begin{lemma}\label{lemma1}
   Let $\bC\in\R^{n\times n}$ be an invertible matrix, and let $\x\in\R^n$ and $\y\in\R^n$ be two vectors. Then 
   $$\x^T\bC^{-1}\y=\sum_{i=1}^n\frac{\text{\normalfont{det}}\,\bC_i}{\text{\normalfont{det}}\,\bC}x_i=\sum_{i=1}^n\frac{\text{\normalfont{det}}\,\widetilde{\bC}_i}{\text{\normalfont{det}}\,\bC}y_i,$$
   where $\bC_i$ is the matrix formed  by replacing the $i$-th column of $\bC$ with $\y$, and $\widetilde{\bC}_i$ is formed by replacing the $i$-th row of $\bC$ with $\x^T$.
   \end{lemma}
  \begin{proof}
  Let $\bC^{-1}\y=\v$. Then $\x^T\bC^{-1}\y=\sum_{i=1}^nv_ix_i.$ According to the Cramer's rule \cite{strang2022introduction},
  $$v_i=\frac{\text{det}\,\bC_i}{\text{det}\,\bC},~i=1,2,\dots, n,$$
  which gives the first equation. The second equation can be proved similarly.

    \end{proof} 
    \begin{lemma}\label{lemma:relationab}
    The following equation holds for the vectors $\{\ba_n\}_{n=1}^{N}$ and $\{\bb_n\}_{n=1}^N$ given in \eqref{def:ab}:
    $$\sum_{n=1}^N\left({\ba}_n\bb_n^T-\bb_n\ba_n^T\right)=\mathbf{0}.$$
    \end{lemma}
    \begin{proof}
Recalling the definitions of $\{\ba_n\}_{n=1}^{N}$ and $\{\bb_n\}_{n=1}^{N}$ in \eqref{def:ab}, it suffices to show that
\begin{equation}\label{Mgamma}
\mathbf{M}_\bU^T\boldsymbol{\Gamma}_{\bU}=\boldsymbol{\Gamma}_{\bU}^T\mathbf{M}_{\bU},
\end{equation}
where 
$\bM_\bU$ and $\boldsymbol{\Gamma}_{\bU}$ are given in \eqref{Mu} and \eqref{gammaU}, respectively.  Since $\bU\in\mathcal{\bar{U}}$, we have
\begin{equation}\label{relation}
(\bU)^{H}\bU=\bH_r\bH_r^H,~~~(\bU)^{T}\bH_r^H=\left((\bU)^{T}\bH_r^H\right)^T,
\end{equation}
which are equivalent to
{\begin{subequations}\label{relation:UH}
\begin{align}
&\RR(\bU)^T\RR(\bU)+\I(\bU)^T\I(\bU)
=\RR(\bH_r)\RR(\bH_r)^T+\I(\bH_r)\I(\bH_r)^T,\\
&\RR(\bU)^T\I(\bU)-\I(\bU)^T\RR(\bU)
=\I(\bH_r)\RR(\bH_r)^T-\RR(\bH_r)\I(\bH_r)^T,\\
&\RR(\bH_r)\RR(\bU)+\I(\bH_r)\I(\bU)
=(\RR(\bH_r)\RR(\bU)+\I(\bH_r)\I(\bU))^{T},\\
&\RR(\bH_r)\I(\bU)-\I(\bH_r)\RR(\bU)
=(\RR(\bH_r)\I(\bU)-\I(\bH_r)\RR(\bU))^{T}.
\end{align}
\end{subequations}}
\hspace{-0.1cm}With \eqref{relation:UH}, it is straightforward to check that \eqref{Mgamma} holds, thereby completing the proof.
\end{proof}
\subsection{Proof of Lemma \ref{lemmaAb}}\label{proof:lemmaAb}
Now we are ready to prove Lemma \ref{lemmaAb}. First, by Lemma \ref{lemmaA}, we have $$\text{rank}([\bA,\bb])\geq \text{rank}(\bA)=N\kappa-\frac{\kappa(\kappa-1)}{2}.$$ 
To complete the proof, we will next show that, with appropriate row transformations,   $\kappa(\kappa-1)/2$ rows of $[\bA~\bb]$ can be made zero, which  implies 
$$\text{rank}([\bA, \bb])\leq N\kappa-\frac{\kappa(\kappa-1)}{2}.$$

Specifically, let $\mathbf{r}_1,\mathbf{r}_2,\dots,\mathbf{r}_{N\kappa}$\footnote{With a slight abuse of notation, we use bold uppercase letters to represent row vectors here.} denote the rows of the matrix $[\bA,\bb]$.   We claim that for all $1\leq j<l\leq \kappa$, 
\begin{equation}\label{rowto0}
\begin{aligned}
\sum_{n=1}^{N-j+1}\alpha_{n,j}&\left(\mathbf{r}_{\kappa(n-1)+l}-
\mathbf{c}_{n,j,l}^T\left[\begin{matrix}
 \mathbf{r}_{\kappa(n-1)+1}\\ \mathbf{r}_{\kappa(n-1)+2}\\\vdots\\\mathbf{r}_{\kappa(n-1)+j}
  \end{matrix}\right]\right)=\mathbf{0},
  \end{aligned}
  \end{equation}
where $\alpha_{n,j}$ and $\mathbf{c}_{n,j,l}$ are given by
\begin{equation}\label{def:c}
\alpha_{n,j}=\left\{
\begin{aligned}
a_{n,1},\hspace{9.55cm}~&\text{if }j=1;\\
a_{n,1}-
\left[\begin{matrix}a_{N-j+2,1}&\cdots&a_{N,1}\end{matrix}\right]
\left[\begin{matrix}a_{N-j+2,2}&\cdots&a_{N,2}\\ \vdots&\ddots&\vdots\\a_{N-j+2,j}&\cdots&a_{N,j}\end{matrix}\right]^{-1}
\left[\begin{matrix}a_{n,2}\\\vdots\\a_{n,j}\end{matrix}\right],\,~~~&\text{otherwise}.
\end{aligned}\right.
\end{equation} 
\vspace{0.1cm}
\begin{equation}\label{def:C}
\mathbf{c}_{n,j,l}=\left\{
\begin{aligned}
\frac{a_{n,1}}{a_{n,1+l}},\hspace{8cm}&\text{if }j=1;\\
\left[\begin{matrix} a_{n,1}&a_{n,2}&\cdots&a_{n,j}\\
a_{N-j+2,1}&a_{N-j+2,2}&\cdots&a_{N-j+2,j}\\
 \vdots&\vdots&\ddots&\vdots\\
  a_{N,1}&a_{N,2}&\cdots&a_{N,j}
  \end{matrix}\right]^{-1}\left[\begin{matrix}a_{n,l}\\a_{N-j+2,l}\\\vdots\\a_{N,l}\end{matrix}\right],~~~~&\text{otherwise},
  \end{aligned}
  \right.~
  \end{equation}
  with $a_{n,l}$ representing the $l$-th element of $\ba_n$, i.e., $a_{n,l}=\ba_n(l)$.
This illustrates that with proper row transformations, we can make the following rows of $[\bA,\bb]$ as  zero vectors:
$$\mathbf{r}_{\kappa(N-j)+l},~~\forall~1\leq j<l\leq \kappa,$$
 whose number is $\kappa(\kappa-1)/2$ in total.
 
  To prove the above claim, we fix $j$ and $l$ (where $1\leq j<l\leq \kappa$),  and denote the vectors in the summation of \eqref{rowto0} as $\hat{\mathbf{r}}_1, \hat{\mathbf{r}}_2, \dots, \hat{\mathbf{r}}_{N-j+1}$, i.e., we write \eqref{rowto0} as $\sum_{n=1}^{N-j+1}\hat{\br}_n=\mathbf{0}$.
   Let $\hat{\mathbf{R}}\in\mathbb{R}^{(N-j+1)\times (N\kappa-\frac{\kappa(\kappa-1)}{2}+1)}$ be a matrix whose rows are  $\hat{\mathbf{r}}_1, \hat{\mathbf{r}}_2, \dots, \hat{\mathbf{r}}_{N-j+1}$, and split $
 \hat{\mathbf{R}}$ as $\hat{\mathbf{R}}=[{\bR}~{\mathbf{d}}],$ where
 $$\begin{aligned}
 {\bR}&=\hat{\mathbf{R}}\left(:,1:N\kappa-\frac{\kappa(\kappa-1)}{2}\right) ~\text{ and }~
 {\bd}=\hat{\mathbf{R}}\left(:,N\kappa-\frac{\kappa(\kappa-1)}{2}+1\right),
 \end{aligned}$$ 
 i.e., $\bR$ is a submatrix of $\hat{\bR}$ with the last column removed, and $\bd$ is the last column of $\hat{\bR}$.
In the following, we prove that  \begin{equation}\label{sumR}\sum_{n=1}^{N-j+1}{\bR}(n,:)=\mathbf{0}\end{equation} and \begin{equation}\label{sumq}\sum_{n=1}^{N-j+1}d_n=0\end{equation} separately. 
\subsubsection{Proof of \eqref{sumR}}To better illustrate the structure of $\bR$, we split it as $\bR=[\bR_1,\bR_2,\dots,\bR_N]$ with $\bR_n\in\R^{(N-j+1)\times \min\{\kappa, N-n+1\}}$. According to the structure of $\bA$ in \eqref{splitA} and the way that $\bR$ is generated, $\bR_n=\mathbf{0}$ if $n>N-j+1$; otherwise 
\begin{equation}\label{Ri}
\begin{aligned}
\bR_n(i,q)\hspace{-0.1cm}=\hspace{-0.1cm}\left\{
\begin{aligned}
&\alpha_{i,j}\hspace{-0.1cm}\left(a_{q+i-1,l}-\bc_{i,j,l}^T\hspace{-0.1cm}\left[\begin{matrix}a_{q+i-1,1}\\\vdots \\a_{q+i-1,j}\end{matrix}\right]\right),~~\text{if }i=n;\\
&\alpha_{i,j}\hspace{-0.1cm}\left(a_{n,l}-\bc_{i,j,l}^T\left[\begin{matrix}a_{n,1}\\\vdots\\a_{n,j}\end{matrix}\right]\right)\hspace{-0.05cm},\hspace{1.6cm}\text{if }i=q+n-1;\\
&0,\hspace{5.6cm}~\text{otherwise},\\
\end{aligned}\right.\\
\text{ where }1\leq i\leq N-j+1,~1\leq q\leq \min\{\kappa,N-n+1\}.
\end{aligned}
\end{equation}
 Due to \eqref{Ri}, there are at most two elements that  are possibly non-zero in each column of $\bR$.   For clarity, we give the structure of $\bR$ in Fig. \ref{fig:R}, where all the non-zero candidates are marked as ``$*$''.  Our goal is to show that: (a) for those columns with only a single ``non-zero candidate'' (a single ``$*$'' in Fig. \ref{fig:R}), its value  is  exactly zero. (b) For those columns with two ``non-zero candidates'', the sum of these two elements is zero\footnote{Here and after,  we only prove the case for $j>1$. The proof for $j=1$ is straightforward.}. 

 \begin{figure}[t]
\includegraphics[scale=0.3]{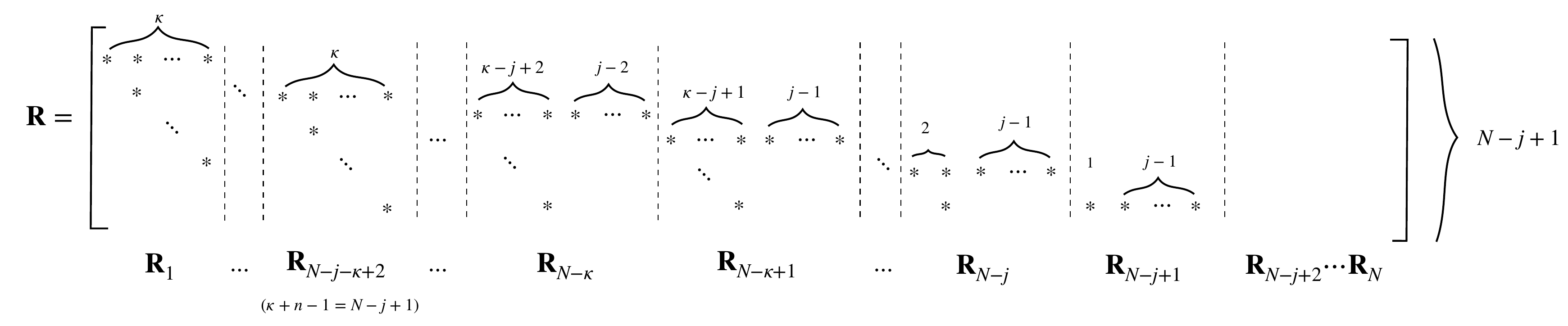}
\centering
\caption{An illustration of the structure of $\mathbf{R}$, where  ``$*$'' marks all the non-zero candidates in ${\mathbf{R}}$ that require further investigation and all the other elements are zero according to \eqref{Ri}.  }
\label{fig:R}
\end{figure}
\hspace{-0.5cm}

\underline{\emph{Proof of (a).}}  We need to show that 
\begin{equation}\label{Rnnq}
\begin{aligned}
\bR_n(n,q)=0,~&\text{if }q=1,~1\leq n\leq N-j+1\\
&\text{or }N-j+2\leq q+n-1\leq N.
\end{aligned}
\end{equation}
First, we note that $(n,q)$ defined in \eqref{Rnnq} satisfies 
\begin{equation}\label{q+n-1}q+n-1\in\{n,N-j+2,N-j+3,\dots,N\}.\end{equation}
Let  $\x=\left[a_{n,l},a_{N-j+2,l},\dots,a_{N,l}\right]^T$\hspace{-0.2cm},  $\y=[a_{q+n-1,1},a_{q+n-1,2},\dots,a_{q+n-1,j}]^T,$ and 
$$~\bC=\left[\begin{matrix} a_{n,1}&a_{N-j+2,1}&\dots&a_{N,1}\\
 a_{n,2}&a_{N-j+2,2}&\dots&a_{N,2}\\
 \vdots&\vdots&\ddots&\vdots\\
  a_{n,j}&a_{N-j+2,j}&\dots&a_{N,j}
  \end{matrix}\right].$$ 
  According to \eqref{Ri} and the definition of $\bc_{i,j,l}$ in \eqref{def:C},   we have
$$
\begin{aligned}
\bR_{n}(n,q)=\alpha_{n,j}\left(a_{q+n-1,l}-\x^T\bC^{-1}\y\right)&=0,
&\forall~(n,q)~\text{satisfies}~\,\eqref{q+n-1}.
\end{aligned}$$
  where the second equality applies the first equation of Lemma \ref{lemma1} and uses the fact that $\y$ is a column of $\bC$ due to \eqref{q+n-1}.  This  completes the proof for (a).

\underline{ \emph{Proof of (b)}}. To prove (b), we need to show that 
 $$
 \begin{aligned}\bR_{n}(n,q)&+\bR_n(q+n-1,q)=0,
 ~&\forall~(q,n): n<q+n-1\leq N-j+1. \end{aligned}$$ 
 For notational simplicity, let   $  \gamma_1=a_{n,l},~\gamma_2=a_{q+n-1,l},~\eta_1=a_{n,1},~\eta_2=a_{q+n-1,1},$
  $ \boldsymbol{\xi}_1=[a_{n,2},\cdots,a_{n,j}]^T,~\boldsymbol{\xi}_2=[a_{q+n-1,2}, \cdots, a_{q+n-1,j}]^T,$
  $ \mathbf{g}=[{a}_{N-j+2,l}, \cdots, a_{N,l}]^T,~\mathbf{f}=[a_{N-j+2,1}, \cdots, a_{N,1}]^T,$ and 
  \begin{equation}\label{D}
  \bD=\left[\begin{matrix}a_{N-j+2,2}&\cdots&a_{N,2}\\ \vdots&\ddots&\vdots\\ a_{N-j+2,j}&\cdots& a_{N,j} \end{matrix}\right].
  \end{equation}

 According to \eqref{def:c}, \eqref{def:C}, and \eqref{Ri},
 $$
  \begin{aligned}
 {\bR}_n(n,q)+{\bR}_n(q+n-1,q)
   =\,&(\eta_1-\mathbf{f}^T\bD^{-1}\boldsymbol{\xi}_1)(\gamma_2-[\gamma_1~~\mathbf{g}^T]\left[\begin{matrix}\eta_1&\mathbf{f}^T\\\boldsymbol{\xi}_1&\bD\end{matrix}\right]^{-1}\left[\begin{matrix}\eta_2\\\boldsymbol{\xi}_2\end{matrix}\right])\\
   &+(\eta_2-\mathbf{f}^T\bD^{-1}\boldsymbol{\xi}_2)(\gamma_1-[\gamma_2~~\mathbf{g}^T]\left[\begin{matrix}\eta_2&\mathbf{f}^T\\\boldsymbol{\xi}_2&\bD\end{matrix}\right]^{-1}\left[\begin{matrix}\eta_1\\\boldsymbol{\xi}_1\end{matrix}\right]).
   \end{aligned}
$$
      We first investigate $\bR_n(n,q)$. Note that 
   $$\left[\begin{matrix}\eta_1&\mathbf{f}^T\\\boldsymbol{\xi}_1&\bD\end{matrix}\right]^{-1}=\left[\begin{matrix}\frac{1}{\eta_1}+\frac{1}{\eta_1^2}\mathbf{f}^T\mathbf{S}\boldsymbol{\xi}_1&-\frac{1}{\eta_1}\mathbf{f}^T\mathbf{S}\\-\frac{1}{\eta_1}\mathbf{S}\boldsymbol{\xi}_1&\mathbf{S}\end{matrix}\right],$$
   where $\mathbf{S}=\left(\bD-\frac{1}{\eta_1}\boldsymbol{\xi}_1\mathbf{f}^T\right)^{-1}$  is the inverse of the Schur complement of $\eta_1$, which, according to the Sherman-Morrison-Woodbury formula, can further be expressed as 
   $$
   \mathbf{S}=\bD^{-1}+\frac{1}{\eta_1-\mathbf{f}^T\bD^{-1}\boldsymbol{\xi}_1}\bD^{-1}\boldsymbol{\xi}_1\mathbf{f}^T\bD^{-1}.$$   Using the above equations, we have 
$$  \begin{aligned}
 \gamma_2-[\gamma_1~~\mathbf{g}^T]\left[\begin{matrix}\eta_1&\mathbf{f}^T\\\boldsymbol{\xi}_1&\bD\end{matrix}\right]^{-1}\left[\begin{matrix}\eta_2\\\boldsymbol{\xi}_2\end{matrix}\right]
=\,& \gamma_2\hspace{-0.05cm}-\hspace{-0.05cm}\frac{\eta_2\gamma_1}{\eta_1}\hspace{-0.05cm}-\hspace{-0.05cm}\frac{\eta_2\gamma_1}{\eta_1^2} \mathbf{f}^T\mathbf{S}\boldsymbol{\xi}_1+\frac{\eta_2}{\eta_1}\mathbf{g}^T\mathbf{S}\boldsymbol{\xi}_1\hspace{-0.03cm}+\hspace{-0.03cm}\frac{\gamma_1}{\eta_1}\mathbf{f}^T\mathbf{S}\boldsymbol{\xi}_2\hspace{-0.03cm}-\hspace{-0.03cm}\mathbf{g}^T\mathbf{S}\boldsymbol{\xi}_2\\
=\,&\gamma_2-\frac{\eta_2\gamma_1}{\eta_1}-\frac{\eta_2\gamma_1}{\eta_1} \frac{\mathbf{f}^T\mathbf{D}^{-1}\boldsymbol{\xi}_1}{\eta_1-\mathbf{f}^T\mathbf{D}^{-1}\boldsymbol{\xi}_1}+\frac{\eta_2\mathbf{g}^T\mathbf{D}^{-1}\boldsymbol{\xi}_1}{\eta_1-\mathbf{f}^T\mathbf{D}^{-1}\boldsymbol{\xi}_1}\\
&+\frac{\gamma_1\mathbf{f}^T\mathbf{D}^{-1}\boldsymbol{\xi}_2}{\eta_1-\mathbf{f}^T\mathbf{D}^{-1}\boldsymbol{\xi}_1}\hspace{-0.05cm}-\hspace{-0.05cm}\mathbf{g}^T\mathbf{D}^{-1}\boldsymbol{\xi}_2\hspace{-0.05cm}-\hspace{-0.05cm}\frac{\left(\mathbf{g}^T\mathbf{D}^{-1}\boldsymbol{\xi}_1\right)\left(\mathbf{f}^T\mathbf{D}^{-1}\boldsymbol{\xi}_2\right)}{\eta_1-\mathbf{f}^T\mathbf{D}^{-1}\boldsymbol{\xi}_1}\\
=\,&\gamma_2-\mathbf{g}^T\mathbf{D}^{-1}\boldsymbol{\xi}_2-(\gamma_1-\mathbf{g}^T\mathbf{D}^{-1}\boldsymbol{\xi}_1)\frac{\eta_2-\mathbf{f}^T\bD^{-1}\boldsymbol{\xi}_2}{\eta_1-\mathbf{f}^T\bD^{-1}\boldsymbol{\xi}_1},
 \end{aligned}$$
and hence 
 $$
 \begin{aligned}
 \bR_{n}(n,q)=&\left(\gamma_2-\mathbf{g}^T\mathbf{D}^{-1}\boldsymbol{\xi}_2\right)\left(\eta_1-\mathbf{f}^T\bD^{-1}\boldsymbol{\xi}_1\right)
 -(\gamma_1-\mathbf{g}^T\mathbf{D}^{-1}\boldsymbol{\xi}_1){(\eta_2-\mathbf{f}^T\bD^{-1}\boldsymbol{\xi}_2)}.
\end{aligned}$$
Similarly, we can show that 
$$
\begin{aligned}
\bR_n(q+n-1,q)=&\left(\gamma_1-\mathbf{g}^T\mathbf{D}^{-1}\boldsymbol{\xi}_1\right)\left(\eta_2-\mathbf{f}^T\bD^{-1}\boldsymbol{\xi}_2\right)
-(\gamma_2-\mathbf{g}^T\mathbf{D}^{-1}\boldsymbol{\xi}_2){(\eta_1-\mathbf{f}^T\bD^{-1}\boldsymbol{\xi}_1)}.
\end{aligned}$$
Combining the above two equations gives the desired result and completes the proof of (b) as well as the proof of \eqref{sumR}.
\subsubsection{Proof of \eqref{sumq}}
 For notational simplicity, we assume  $l=j+1$ in the following proof. This assumption is made without loss of generality and the following proof can be extended straightforwardly to any $l=j+2,\dots, \kappa$. 
 
 With $l=j+1$, $d_n$ is given by 
\begin{equation}\label{3}
\begin{aligned}
{d}_n
=\,&\alpha_{n,j}\left(b_{n,j+1}-[a_{n,j+1},a_{N-j+2,j+1},\dots, a_{N,j+1}]\,\left[\begin{matrix} a_{n,1}&a_{N-j+2,1}&\dots&a_{N,1}\\
 a_{n,2}&a_{N-j+2,2}&\dots&a_{N,2}\\
 \vdots&\vdots&\ddots&\vdots\\
  a_{n,j}&a_{N-j+2,j}&\dots&a_{N,j}
  \end{matrix}\right]^{-1}\left[\begin{matrix}
b_{n,1}\\b_{n,2}\\\vdots\\b_{n,j}
  \end{matrix}\right]\right).\\
  \end{aligned}
  \end{equation}
Define  
\begin{equation}\label{Phi}
\boldsymbol{\Phi}:=\left[\begin{matrix} a_{N-j+2,1}&\dots&a_{N,1}\\
a_{N-j+2,2}&\dots&a_{N,2}\\
\vdots&\ddots&\vdots\\
a_{N-j+2,j}&\dots&a_{N,j}\\
a_{N-j+2,j+1}&\dots& a_{N,j+1}
  \end{matrix}\right]\in\mathbb{R}^{(j+1)\times (j-1)}. 
  \end{equation}
In addition, denote by  $\boldsymbol{\Phi}^{\backslash s,q}$ the submatrix of $\boldsymbol{\Phi}$ with the $s$-th and $q$-th rows removed. With this notation, $\bD$ given in \eqref{D} can be written as $\mathbf{D}=\boldsymbol{\Phi}^{\backslash 1,j+1}$.
We claim that 
{\small   $$
   \begin{aligned}
  \sum_{n=1}^{N-j+1} d_n
   \overset{(a)}{=}&\frac{1}{\text{det} \,\bD}\sum_{q=1}^{j+1}\sum_{s=1}^{q-1}\sum_{n=1}^{N-j+1}(-1)^{j-q+s}\text{det}(\boldsymbol{\Phi}^{\backslash s,q})a_{n,s}b_{n,q}\\
   &+\frac{1}{\text{det} \,\bD}\sum_{q=1}^{j+1}\sum_{s=q+1}^{j+1}\sum_{n=1}^{N-j+1}(-1)^{j-q+s-1}\text{det}(\boldsymbol{\Phi}^{\backslash s,q})a_{n,s}b_{n,q}\\
   \overset{(b)}{=}&\frac{1}{\text{det} \,\bD}\sum_{q=1}^{j+1}\sum_{s=1}^{q-1}(-1)^{j-q+s}\text{det}(\boldsymbol{\Phi}^{\backslash s,q})\sum_{n=1}^N(a_{n,s}b_{n,q}-a_{n,q}b_{n,s})\\&-\frac{1}{\text{det} \,\bD}\sum_{q=1}^{j+1}\sum_{s=1}^{q-1}(-1)^{j-q+s}\text{det}(\boldsymbol{\Phi}^{\backslash s.q})\hspace{-0.3cm}\sum_{n=N-j+2}^N\hspace{-0.3cm}(a_{n,s}b_{n,q}\hspace{-0.05cm}-\hspace{-0.05cm}a_{n,q}b_{n,s})\\
   \overset{(c)}{=}&-\frac{1}{\text{det} \,\bD}\sum_{q=1}^{j+1}\sum_{s=1}^{q-1}(-1)^{j-q+s}\text{det}(\boldsymbol{\Phi}^{\backslash s,q})\hspace{-0.3cm}\sum_{n=N-j+2}^N\hspace{-0.3cm}(a_{n,s}b_{n,q}\hspace{-0.05cm}-\hspace{-0.05cm}a_{n,q}b_{n,s})\\
   \overset{(d)}{=}&\,0,
   \end{aligned}
   $$   }
 \hspace{-0.27cm} where (b) follows immediately by exchanging indices $s$ and $q$ in the second term of the r.h.s. of the first equation, and (c) holds since $\sum_{n=1}^N(a_{n,s}b_{n,q}-a_{n,q}b_{n,s})=0$ due to Lemma \ref{lemma:relationab}.
    We next prove (a) and (d) separately.
   
\underline{\emph{Proof of (a).}} Applying the second equation of Lemma \ref{lemma1} with $\x=[a_{n,j+1},a_{N-j+2,j+1},\dots, a_{N,j+1}]^T$, $\y=[b_{n,1},b_{n,2},\dots,b_{n,j}]^T$, and  
   \begin{equation}\label{C}
   {\bC}=\left[\begin{matrix} a_{n,1}&a_{N-j+2,1}&\dots&a_{N,1}\\
 a_{n,2}&a_{N-j+2,2}&\dots&a_{N,2}\\
 \vdots&\vdots&\ddots&\vdots\\
  a_{n,j}&a_{N-j+2,j}&\dots&a_{N,j}
  \end{matrix}\right],
  \end{equation}
we have 
$$d_n=\alpha_{n,j}b_{n,j+1}-\alpha_{n,j}\sum_{q=1}^j\frac{\text{det}\,\bC_q}{\text{det}\, \bC}b_{n,q},$$ where $\bC_{q}$ is the matrix obtained by replacing the $q$-th row of $\bC$ with $[a_{n,j+1},a_{N-j+2,j+1},\dots, a_{N,j+1}]$. In addition, according to the definitions of $\alpha_{n,j}$ in \eqref{def:c}, $\bD$ in \eqref{D},  and $\bC$ in \eqref{C}, it is easy to check that $\alpha_{n,j}$ is the Schur complement of block $\bD$ of matrix $\bC$. Therefore, we have $\text{det}\,\bC=\alpha_{n,j}\,\text{det}\,\bD$, and thus 
      \begin{equation}\label{10}
    {d}_n=\frac{\text{det}\,\bC}{\text{det}\,\bD}\,b_{n,j+1}-\sum_{q=1}^{j}\frac{\text{det}\,\bC_{q}}{\text{det}\,\bD}\,b_{n,q}.   
  \end{equation}
  For the first term, applying the Laplace's expansion theorem \cite{strang2022introduction} to the first column of $\bC$, we get 
  \begin{equation}\label{firstterm}
  \frac{\text{det}\,\bC}{\text{det}\,\bD}\,b_{n,j+1}=\frac{1}{\text{det}\,\bD}\sum_{s=1}^{j}(-1)^{s+1}\text{det}\,\boldsymbol{\Phi}^{\backslash s,j+1}a_{n,s}b_{n,j+1},
  \end{equation}
  where it is easy to check that $\text{det}\,\boldsymbol{\Phi}^{\backslash s,j+1}$ is the  minor of $a_{n,s}$ for $\bC$.
  We next investigate $\frac{\text{det}\,\bC_{q}}{\text{det}\,\bD}\,b_{n,q},~q=1,2,\dots, j$. Let $\tilde{\bC}_{q}$ be the matrix obtained by rotating the $q$-th row of $\bC_{q}$, i.e., $[a_{n,j+1},\dots, a_{N,j+1}]$, to the last row, i.e., 
  $$\tilde{\bC}_q=\left[\begin{matrix} a_{n,1}&a_{N-j+2,1}&\cdots&a_{N,1}\\\vdots&\vdots&\ddots&\vdots\\a_{n,q-1}&a_{N-j+2,q-1}&\cdots&a_{N,q-1}\\a_{n,q+1}&a_{N-j+2,q+1}&\cdots&a_{N,q+1}\\\vdots&\vdots&\ddots&\vdots\\a_{n,j+1}&a_{N-j+2,j+1}&\cdots&a_{N,j+1}
  \end{matrix}\right].$$ Since it requires  a number of  $j-q$ row rotations to obtain $\tilde{\bC}_{q}$ from $\bC_{q}$, we have 
  \begin{equation}\label{11}
\text{det}\,\bC_{q}=(-1)^{j-q}\,\text{det}\,\tilde{\bC}_{q}.
  \end{equation} 
   Note that $a_{n,s}$ is the $(s,1)$-th element of  $\tilde{\bC}_{q}$  if $s<q$ and the $(s-1,1)$-th element of  $\tilde{\bC}_{q}$  if $s>q$. Therefore, 
  \begin{equation}\label{12}
  \begin{aligned}
  \frac{\text{det}\,\bC_{q}}{\text{det}\,\bD}\,b_{n,q}&=\frac{(-1)^{j-q}\,\text{det}\,\tilde{\bC}_{q}}{\text{det}\bD}b_{n,q}\\
  &=\sum_{s=1}^{q-1}\frac{(-1)^{j-q+s+1}}{\text{det}\,\bD}\text{det}\,\boldsymbol{\Phi}^{\backslash s,q}a_{n,s}b_{n,q}
  +\sum_{s=q+1}^{j+1}\frac{(-1)^{j-q+s}}{\text{det}\,\bD}\text{det}\boldsymbol{\Phi}^{\backslash s,q}a_{n,s}b_{n,q},
  \end{aligned}
  \end{equation}
  where the second equation is obtained by applying the Laplace's expansion theorem to the first column of $\tilde{\bC}_q$  and by noting that $\text{det}\,\boldsymbol{\Phi}^{\backslash s,q}$ is the minor of $a_{n,s}$ for $\tilde{\bC}_{q}$. Combining \eqref{10}, \eqref{firstterm}, and \eqref{12}, we have 
    $$
    \begin{aligned}d_n=&\frac{1}{\text{det}\,\bD}\sum_{q=1}^{j+1}\sum_{s=1}^{q-1}(-1)^{j-q+s}\text{det}(\boldsymbol{\Phi}^{\backslash s,q})a_{n,s}b_{n,q}
    +\frac{1}{\text{det}\,\bD}\sum_{q=1}^{j+1}\sum_{s=q+1}^{j+1}(-1)^{j-q+s+1}\text{det}(\boldsymbol{\Phi}^{\backslash s,q})a_{n,s}b_{n,q},
    \end{aligned}
    $$
which completes the  proof for (a).

\underline{\emph{Proof of (d)}}.
 We now prove  
 $$\sum_{q=1}^{j+1}\sum_{s=1}^{q-1}(-1)^{j-q+s}\text{det}(\boldsymbol{\Phi}^{\backslash s,q})\hspace{-0.15cm}\sum_{n=N-j+2}^N\hspace{-0.1cm}(a_{n,s}b_{n,q}-a_{n,q}b_{n,s})=0.$$
  We first rearrange the left hand side of the above equation as follows: 
  $$\begin{aligned}
 & \sum_{q=1}^{j+1}\sum_{s=1}^{q-1}(-1)^{j-q+s}\text{det}(\boldsymbol{\Phi}^{\backslash s,q})\sum_{n=N-j+2}^Na_{n,s}b_{n,q}
 -\sum_{q=1}^{j+1}\sum_{s=1}^{q-1}(-1)^{j-q+s}\text{det}(\boldsymbol{\Phi}^{\backslash s,q})\sum_{n=N-j+2}^Na_{n,q}b_{n,s}\\
  =\,&\sum_{q=1}^{j+1}(-1)^{j-q}\left(\sum_{s=1}^{q-1}(-1)^{s}\text{det}(\boldsymbol{\Phi}^{\backslash s,q})\sum_{n=N-j+2}^Na_{n,s}
  +\sum_{s=q+1}^{j+1}(-1)^{s+1}\text{det}(\boldsymbol{\Phi}^{\backslash s,q})\sum_{n=N-j+2}^Na_{n,s}\right)b_{n,q},
  \end{aligned}
  $$
  where the equality is obtained by exchanging indices $s$ and $q$ of the second term in the first equation. 
 Note that for given $q=1,2,\dots, j+1,$
  $$
  \begin{aligned}
  &\sum_{s=1}^{q-1}(-1)^{s}\text{det}(\boldsymbol{\Phi}^{\backslash s,q})\sum_{n=N-j+2}^Na_{n,s}
  +\sum_{s=q+1}^{j+1}(-1)^{s+1}\text{det}(\boldsymbol{\Phi}^{\backslash s,q})\sum_{n=N-j+2}^Na_{n,s}\\
  =\,&\hspace{-0.05cm}-\hspace{-0.05cm}\text{det}\left(\left[\begin{matrix}
 \sum_{n=N-j+2}^Na_{n,1}\hspace{-0.1cm}&  a_{N-j+2,1}&\hspace{-0.1cm}\dots\hspace{-0.1cm}&a_{N,1}\\
 \vdots\hspace{-0.1cm}& \vdots&\hspace{-0.1cm}\ddots\hspace{-0.1cm}&\vdots\\
\sum_{n=N-j+2}^Na_{n,q-1}\hspace{-0.1cm}&a_{N-j+2,q-1}&\hspace{-0.1cm}\dots\hspace{-0.1cm}&a_{N,q-1}\\
\sum_{n=N-j+2}^Na_{n,q+1}\hspace{-0.1cm}&a_{N-j+2,q+1}&\hspace{-0.1cm}\dots\hspace{-0.1cm}&a_{N,q+1}\\
\vdots\hspace{-0.1cm}&\vdots&\hspace{-0.1cm}\ddots\hspace{-0.1cm}&\vdots\\
\sum_{n=N-j+2}^Na_{n,j+1}\hspace{-0.1cm}&a_{N-j+2,j+1}&\hspace{-0.1cm}\dots\hspace{-0.1cm}& a_{N,j+1}
  \end{matrix}\right]\right)\\
 =\,&0,
 \end{aligned}$$
  which gives the desired result. 
  \bibliographystyle{IEEEtran}
\bibliography{IEEEabrv,BDRIS}

\begin{thebibliography}{10}
\providecommand{\url}[1]{#1}
\csname url@samestyle\endcsname
\providecommand{\newblock}{\relax}
\providecommand{\bibinfo}[2]{#2}
\providecommand{\BIBentrySTDinterwordspacing}{\spaceskip=0pt\relax}
\providecommand{\BIBentryALTinterwordstretchfactor}{4}
\providecommand{\BIBentryALTinterwordspacing}{\spaceskip=\fontdimen2\font plus
\BIBentryALTinterwordstretchfactor\fontdimen3\font minus
  \fontdimen4\font\relax}
\providecommand{\BIBforeignlanguage}[2]{{%
\expandafter\ifx\csname l@#1\endcsname\relax
\typeout{** WARNING: IEEEtran.bst: No hyphenation pattern has been}%
\typeout{** loaded for the language `#1'. Using the pattern for}%
\typeout{** the default language instead.}%
\else
\language=\csname l@#1\endcsname
\fi
#2}}
\providecommand{\BIBdecl}{\relax}
\BIBdecl

\bibitem{6G}
W.~Saad, M.~Bennis, and M.~Chen, ``A vision of {6G} wireless systems:
  Applications, trends, technologies, and open research problems,'' \emph{IEEE
  Netw.}, vol.~34, no.~3, pp. 134--142, May/Jun. 2020.

\bibitem{RIS1}
Q.~Wu and R.~Zhang, ``Intelligent reflecting surface enhanced wireless network
  via joint active and passive beamforming,'' \emph{IEEE Trans. Wireless
  Commun.}, vol.~18, no.~11, pp. 5394--5409, Nov. 2019.

\bibitem{RIS_survey}
M.~Di~Renzo, A.~Zappone, M.~Debbah, M.-S. Alouini, C.~Yuen, J.~de~Rosny, and
  S.~Tretyakov, ``Smart radio environments empowered by reconfigurable
  intelligent surfaces: How it works, state of research, and the road ahead,''
  \emph{IEEE J. Sel. Areas Commun.}, vol.~38, no.~11, pp. 2450--2525, Nov.
  2020.

\bibitem{RIS_mag}
E.~Bj\"{o}rnson, O.~\"{O}zdogan, and E.~G. Larsson, ``Reconfigurable
  intelligent surfaces: Three myths and two critical questions,'' \emph{IEEE
  Commun. Mag.}, vol.~58, no.~12, pp. 90--96, Dec. 2020.

\bibitem{BDRIS}
S.~Shen, B.~Clerckx, and R.~Murch, ``Modeling and architecture design of
  reconfigurable intelligent surfaces using scattering parameter network
  analysis,'' \emph{IEEE Trans. Wireless. Commun.}, vol.~21, no.~2, pp.
  1229--1243, Feb. 2022.

\bibitem{BDRIS_survey}
H.~Li, S.~Shen, M.~Nerini, and B.~Clerckx, ``Reconfigurable intelligent
  surfaces {2.0}: Beyond diagonal phase shift matrices,'' \emph{IEEE Commun.
  Mag.}, vol.~62, no.~3, pp. 102--108, Mar. 2024.

\bibitem{closeform}
M.~Nerini, S.~Shen, and B.~Clerckx, ``Closed-form global optimization of beyond
  diagonal reconfigurable intelligent surfaces,'' \emph{IEEE Trans. Wireless
  Commun.}, vol.~23, no.~2, pp. 1037--1051, Feb. 2024.

\bibitem{group_conn}
H.~Li, S.~Shen, and B.~Clerckx, ``Beyond diagonal reconfigurable intelligent
  surfaces: From transmitting and reflecting modes to single-, group-, and
  fully-connected architectures,'' \emph{IEEE Trans. Wireless Commun.},
  vol.~22, no.~4, pp. 2311--2324, Apr. 2023.

\bibitem{coverage}
------, ``Beyond diagonal reconfigurable intelligent surfaces: A multi-sector
  mode enabling highly directional full-space wireless coverage,'' \emph{IEEE
  J. Sel. Areas Commun.}, vol.~41, no.~8, pp. 2446--2460, Aug. 2023.

\bibitem{attack}
H.~Wang, Z.~Han, and A.~L. Swindlehurst, ``Channel reciprocity attacks using
  intelligent surfaces with non-diagonal phase shifts,'' \emph{IEEE Open J.
  Commun. Soc.}, vol.~5, pp. 1469--1485, 2024.

\bibitem{duplex}
\BIBentryALTinterwordspacing
H.~Li and B.~Clerckx, ``Non-reciprocal beyond diagonal {RIS}: Multiport network
  models and performance benefits in full-duplex systems,'' 2024. [Online].
  Available: \url{https://arxiv.org/abs/2411.04370}
\BIBentrySTDinterwordspacing

\bibitem{liu_nonreciprocal}
\BIBentryALTinterwordspacing
Z.~Liu, H.~Li, and B.~Clerckx, ``Non-reciprocal beyond diagonal ris: Sum-rate
  maximization in full-duplex communications,'' 2024. [Online]. Available:
  \url{https://arxiv.org/abs/2411.18523}
\BIBentrySTDinterwordspacing

\bibitem{Linonreciprocal2022}
Q.~Li, M.~El-Hajjar, I.~Hemadeh, A.~Shojaeifard, A.~A.~M. Mourad, B.~Clerckx,
  and L.~Hanzo, ``Reconfigurable intelligent surfaces relying on non-diagonal
  phase shift matrices,'' \emph{IEEE Trans. Veh. Technol.}, vol.~71, no.~6, pp.
  6367--6383, Jun. 2022.

\bibitem{grouping}
H.~Li, S.~Shen, and B.~Clerckx, ``A dynamic grouping strategy for beyond
  diagonal reconfigurable intelligent surfaces with hybrid transmitting and
  reflecting mode,'' \emph{IEEE Trans. Vehicular Tech.}, vol.~72, no.~12, pp.
  16\,748--16\,753, Dec. 2023.

\bibitem{grouping2}
M.~Nerini, S.~Shen, and B.~Clerckx, ``Static grouping strategy design for
  beyond diagonal reconfigurable intelligent surfaces,'' \emph{IEEE Commun.
  Lett.}, vol.~28, no.~7, pp. 1708--1712, Jul. 2024.

\bibitem{tree}
M.~Nerini, S.~Shen, H.~Li, and B.~Clerckx, ``Beyond diagonal reconfigurable
  intelligent surfaces utilizing graph theory: Modeling, architecture design,
  and optimization,'' \emph{IEEE Trans. Wireless Commun.}, vol.~23, no.~8, pp.
  9972--9985, Aug. 2024.

\bibitem{forest}
M.~Nerini and B.~Clerckx, ``Pareto frontier for the performance-complexity
  trade-off in beyond diagonal reconfigurable intelligent surfaces,''
  \emph{IEEE Commun. Lett.}, vol.~27, no.~10, pp. 2842--2846, Oct. 2023.

\bibitem{wu}
\BIBentryALTinterwordspacing
Z.~Wu and B.~Clerckx, ``Optimization of beyond diagonal {RIS}: A universal
  framework applicable to arbitrary architectures,'' 2024. [Online]. Available:
  \url{https://arxiv.org/abs/2412.15965}
\BIBentrySTDinterwordspacing

\bibitem{qstem}
\BIBentryALTinterwordspacing
X.~Zhou, T.~Fang, and Y.~Mao, ``A novel {Q}-stem connected architecture for
  beyond-diagonal reconfigurable intelligent surfaces,'' 2024. [Online].
  Available: \url{https://arxiv.org/abs/2411.18480}
\BIBentrySTDinterwordspacing

\bibitem{PDD}
Y.~Zhou, Y.~Liu, H.~Li, Q.~Wu, S.~Shen, and B.~Clerckx, ``Optimizing power
  consumption, energy efficiency, and sum-rate using beyond diagonal {RIS}—a
  unified approach,'' \emph{IEEE Trans. Wireless Commun.}, vol.~23, no.~7, pp.
  7423--7438, Jul. 2024.

\bibitem{microwavebook}
D.~Pozar, \emph{Microwave Engineering}.\hskip 1em plus 0.5em minus 0.4em\relax
  John Wiley \& Sons, 2011.

\bibitem{generalmodel}
M.~Nerini, S.~Shen, H.~Li, M.~Di~Renzo, and B.~Clerckx, ``A universal framework
  for multiport network analysis of reconfigurable intelligent surfaces,''
  \emph{IEEE Trans. Wireless Commun.}, vol.~23, no.~10, pp. 14\,575--14\,590,
  Oct. 2024.

\bibitem{bondy2008graph}
J.~A. Bondy and U.~S.~R. Murty, \emph{Graph theory}.\hskip 1em plus 0.5em minus
  0.4em\relax Springer, 2008.

\bibitem{DoF}
E.~Telatar, ``Capacity of multi-antenna gaussian channels,'' \emph{European
  Trans. Telecom.}, vol.~10, no.~6, pp. 585--595, 1999.

\bibitem{DoF2}
N.~Jindal and A.~Goldsmith, ``Dirty-paper coding versus {TDMA} for {MIMO}
  broadcast channels,'' \emph{IEEE Trans. Inf. Theory}, vol.~51, no.~5, pp.
  1783--1794, May 2005.

\bibitem{capacity}
H.~Weingarten, Y.~Steinberg, and S.~Shamai, ``The capacity region of the
  {Gaussian} multiple-input multiple-output broadcast channel,'' \emph{IEEE
  Trans. Inf. Theory}, vol.~52, no.~9, pp. 3936--3964, Sept. 2006.

\bibitem{strang2022introduction}
G.~Strang, \emph{Introduction to linear algebra}.\hskip 1em plus 0.5em minus
  0.4em\relax SIAM, 2022.

\end{thebibliography}

\end{document}